\documentclass[10pt, conference]{IEEEtran}

\usepackage[table,xcdraw]{xcolor}

\usepackage{amssymb,amsmath,amsthm,amsfonts}

\usepackage{graphicx}
\graphicspath{{Figures/}}
\usepackage{subfigure}
\usepackage{url}
\usepackage{comment}
\usepackage{multirow}

\usepackage{graphicx}

\usepackage{algorithmic} 

\usepackage[linesnumbered,ruled]{algorithm2e}

\SetAlFnt{\small}
\SetAlCapFnt{\small}
\SetAlCapNameFnt{\small}
\SetAlCapHSkip{0pt}
\IncMargin{-\parindent}

\usepackage{upgreek}
\usepackage{wrapfig}

\usepackage{fixfoot}

\usepackage{listings}
\usepackage{psfrag,wrapfig}
\usepackage{xspace}

\usepackage{paralist} 

\usepackage{enumerate}
\usepackage{subfiles} 

\newcommand{\eg}{{\it e.g.,}\xspace}

\newcommand{\ie}{{\it i.e.,}\xspace}
\newcommand{\etc}{{\it etc.}}
\newcommand{\ci}{{\it (i) }}
\newcommand{\cii}{{\it (ii) }}
\newcommand{\ciii}{{\it (iii) }}

\newcommand{\ca}{{\it (a) }}
\newcommand{\cb}{{\it (b) }}
\newcommand{\cc}{{\it (c) }}
\newcommand{\cd}{{\it (d) }}
\newcommand{\ce}{{\it (e) }}

\newtheorem{theorem}{Theorem}
\newtheorem{definition}[theorem]{Definition}

\begin{document}
\title{Restart-Based Security Mechanisms for Safety-Critical Embedded Systems}

 \author{
 \IEEEauthorblockN{Fardin Abdi, Chien-Ying Chen, Monowar Hasan, Songran Liu, Sibin Mohan and Marco Caccamo} 
 \IEEEauthorblockA{Dept. of Computer Science, University of Illinois at Urbana-Champaign, Urbana, IL, USA}
 }

\maketitle

\begin{abstract}
	\textbf{\small
		Many physical plants that are controlled by embedded systems have safety
		requirements that need to be respected at all times -- any deviations
		from expected behavior can result in damage to the system (often to the
		physical plant), the environment or even endanger human life. In recent
		times, malicious attacks against such systems have increased -- many with the
		intent to cause physical damage. 
		In this paper, we aim to decouple the \textit{safety} of the plant from 
		\textit{security} of the embedded system by taking advantage of the inherent inertia in such systems. 
		In this paper we present a system-wide restart-based framework that combines hardware 
		and software components to \ca maintain the system within the safety region
		and \cb thwart potential attackers from destabilizing the system.
		We demonstrate the feasibility of our approach using two realistic systems -- 
		an actual 3 degree of freedom (3-DoF) helicopter and a simulated warehouse
		temperature control unit.
		Our proof-of-concept implementation is tested against
		multiple emulated attacks on the control units of these systems.
	}
	
\end{abstract}

\section{Introduction}

Conventionally, embedded cyber-physical systems (CPS) used custom platforms, software/protocols and were normally not linked to external networks. As a result, security was typically not a priority in their design.
However, due to the drive towards remote monitoring and control facilitated by the growth of the Internet along with the increasing use of common-off-the-shelf (COTS) components, these traditional assumptions are increasingly being challenged as evidenced by recent attacks \cite{dronhack, ris_rts_1, 5751382}. Any successful, serious attack against such systems could lead to problems more catastrophic than just loss of data or availability because of the critical nature of such systems \cite{checkoway2011comprehensive,gollakota2011they}.

One of the challenges of designing secure embedded systems is to develop techniques/frameworks that consider the safety requirements of any physical system that they may control and operate within resource-constrained environments of embedded platforms. In the design of such a secure framework, it is important to distinguish between \textit{physical safety} and \textit{cyber-security}. For many embedded systems, violation of cyber-security~(\eg stealing system secrets, logging user activity pattern, taking control of the system, \etc) by itself may not be considered as harmful as violation of physical safety~(\eg destabilizing a helicopter by moving towards other objects or ground). 
In fact, in the context of embedded CPS, one of the goals of an adversary could be violate physical safety eventually and cause physical damage. Unfortunately, many of the existing security techniques are designed for traditional computing systems and do not take into the account the physical aspect of embedded systems. Hence, they are not always effective in ensuring physical safety in the presence of attacks.

On one hand, many existing security mechanisms aim to protect against malicious activities and/or detecting them as quickly as possible. It is, however, not easily possible to detect or protect from \textit{all} the attacks without any false negatives. On the other hand, safety guarantees verified by design methods for safety-critical CPS are only valid for the original uncorrupted system software. In other words, \textit{if security of the system can not be guaranteed, safety requirements may also be broken}. 
Hence there is need for security mechanisms that are specifically designed for safety-critical systems with verifiability in mind.

In this paper, we propose a design method for safety-critical embedded CPS that \textit{decouples physical safety} from \textit{cyber security}, \ie it guarantees the safety of the physical system even when the security of the embedded system cannot be guaranteed. 
 The main idea of the paper is based on two key observations. \textit{First,} restarting a computing unit and reloading the software from a protected storage is an effective way to remove any malicious component from the system (albeit for a short, finite amount of time) and restore the software to an uncorrupted, functional condition~\cite{mhasan_resecure16}. \textit{Second,} physical systems, unlike computing systems, have \textit{inertia}. As a result, changing the state of a physical plant from a safe state to an unsafe one even with complete adversarial control\footnote{For instance if the adversary gains root/administrative access in the system.} is not instantaneous and takes considerable time (\eg increasing temperature, increasing pressure, \etc) We leverage
	this property to calculate a \textit{safe operational window} and combine it with the effectiveness of {\em system-wide
		restarts} to decrease the efficacy of malicious actors.

In our approach, the power of system-restarts is used to \textit{recover the embedded system from a potentially corrupted condition to a trusted and functional state}. At this point, all external interfaces~(\eg network and debugging interface) of the embedded system (except for the sensors and actuators) are disabled. While system is in this isolated mode, a safe operational window is calculated during which, even with full control of the adversary, the physical plant cannot reach an unsafe state. In this isolated mode the next restart time of the embedded system is also set dynamically. Then, the main controller and all interfaces are activated and system starts the normal operation until the restart time time is reached. This process repeats after each restart. However, if at some point, the plant is too close to the safety boundary such that calculated safety window is too short~(\ie system cannot tolerate any adversarial behavior), the secure/isolated operational mode is extended. During the extended period,  a \textit{safety controller (SC)} drives the plant towards the safe states and farther away from the unsafe boundary. Once the system has a reasonable margin, next restart time is set and normal controller is activated. In our design, \textit{an external isolated simple hardware module} acts as the \emph{Root of Trust~(RoT)}\footnote{The proposed RoT module can be constructed using simple commercial off-the-shelf~(COTS) components.} and triggers the restarts (Section \ref{sec:RoTDesign}). This guarantees that even very strong adversaries cannot interfere with the system restart events.



Notice that only restarting the embedded platform may make it easier for the attacker to penetrate the system again. However the dynamic calculation of restart times and other quantities (\eg reachability calculations, see Section \ref{sec:securetime}) makes it harder for the attacker to succeed. Along with restart, using the approaches described in Section \ref{sec:discussion}, it is possible to further increase the difficulty for the attacker to intrude and cause damage to the system.

Note that with optimizing the boot sequence and given the embedded nature of such systems, \textit{the time required to restart and reload the software can be very short}. The impact of the proposed approach on availability of the system varies according to the boot time of the particular embedded system~(type of the OS, processor, storage, size of the software, \etc) and the physical plant's dynamics~(that determines the rate of restarts).  Such an impact is negligible for systems with slower dynamics such as temperature control system and it can increase for unstable systems such as avionics (see Section~\ref{sec:evaluation}). Figure~\ref{fig:applicableFields} illustrates the impact of the proposed approach on some categories of applications, in an abstract sense.

\begin{figure}[ht]
	\begin{center}
		\includegraphics[width=0.35\textwidth]{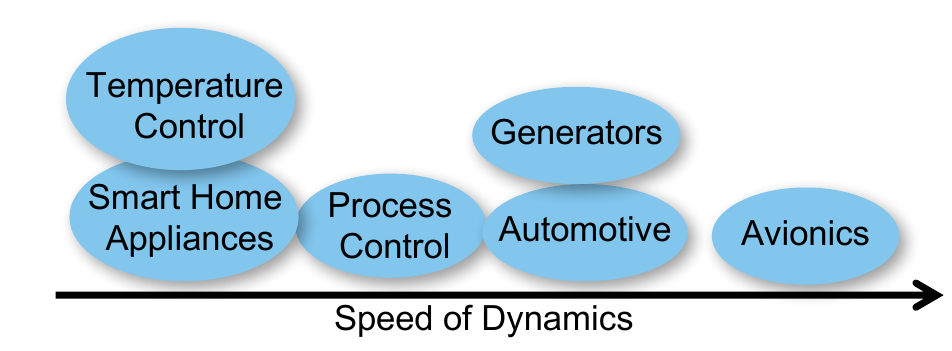}
		\caption{Impact of the proposed restart-based protection increases with the Speed of dynamics of the physical sub-system.}
		\label{fig:applicableFields}
	\end{center}
\end{figure}

In this paper we make following contributions:
\begin{itemize}
	\item We propose an approach to calculate the safe operation window for CPS and hence, improve its overall security. By utilizing system-wide restarts the proposed framework can recover a potentially compromised system. (Section~\ref{sec:probability}).
    \item We guarantee the safety of the physical system despite security threats (Section~\ref{sec:details}).
	\item  We demonstrate the effectiveness of our proposed method using a proof-of-concept implementation as applied to realistic systems (Section~\ref{sec:evaluation}).
\end{itemize}

\section{System and Adversary Model}
\label{sec:model}

\subsection{System Model}

The type of system that we consider in this work is an embedded control system that drives the physical plant. Examples of such systems include heating, ventilation and watering systems, smart home appliances, drone control systems, automotive engine controllers, process controllers, \etc  These systems usually provide an interface to interact with user or other physical systems. These interfaces may be utilized for transmitting the state and sensors values or receiving new set points and operation plans. In some cases, embedded platforms also provide debugging interfaces that can be used for diagnosing errors during operation or the development cycle. Furthermore, the physical plant has \textit{safety requirements} that need to be respected at \textit{all} times. Failure to satisfy the safety requirements may damage the physical components or surrounding environment. As an example, a helicopter/UAV may destabilized (or damaged) if its wings touches the ground or other objects.



\subsection{Adversary Model}

Embedded CPS face threats in various forms depending on the system and the goals of the attacker. For instance, adversaries may insert, eavesdrop on or modify messages exchanged by system components. Besides, attackers may manipulate the processing of sensor inputs and actuator commands, could try to modify the control flow of the system as well as extract sensitive information through side-channels. In this paper we make the following assumptions about the system and capabilities of the adversary:

\begin{enumerate}[\it i\normalfont )]
	
	\item \textit{Software image integrity:} We assume that the original image of the system software~(real-time operating system~(RTOS), control applications and other components) does not contain malicious components. These components, however, may contain bugs or security vulnerabilities that could be exploited to initiate attacks.
	
	\item \textit{Read-only memory unit}: We assume that the original image of the system software is stored on a read-only memory unit (\ie E$^2$PROM). This content is unmodifiable at runtime by adversary or other entities.

	\item \textit{Integrity of Root of Trust (RoT)}: As mentioned earlier, RoT is an isolated hardware module responsible for issuing a restart signal 
    at designated times. It is designed to be directly accessible \textit{only} during an internal that we call the SEI~(refer to Section~\ref{sec:RoTDesign}). Hence, we consider that attackers can not compromise it to prevent the system from restarting. 	
	
	\item \textit{Sensor and actuator integrity}: We assume that the adversary does not have physical access to the sensors and actuators. Hence, values reported by sensors are correct and commands sent to the actuators are executed accordingly. However, an attacker may corrupt and replace sensor measurements or actuator commands inside the applications, on the network or underlying OS.
	
	\item \textit{External attacks}: Attackers require an external interface such as the network, serial port or debugging interface to launch the attacks. As long as all such interfaces remain disabled, applications running on the system are assumed to be trusted. 

	\item \textit{Integrity violation}: We assume that once the external interfaces are activated, the adversary can compromise the software components\footnote{We are not concerned with exact method used by attackers to get into the system.} on the embedded system including the RTOS as well as the real-time/control applications. 

	
\end{enumerate}



Our approach also improves system security against forms of attack that have intentions other than damaging the physical plant. The following are some instances of such attacks \cite{cy_side_channel,mhasan_resecure16}: \ci \textit{Information leakage through side-channels}: The adversary may aim to learn important information through side or covert-channel attacks by simply lodging themselves in the system and extracting sensitive information. \cii \textit{System-level Denial of Service (DoS)}: The attacker may intrude into the real-time system and exhaust \textit{system-level} resources such as CPU, disk, memory, restricting safety-critical control tasks from necessary resources, \etc~ 
\ciii \textit{Control performance degradation}:
An attacker may focus on reducing the control performance and consequently reducing the progress of the system towards the mission goal by tampering with actuators/sensors and/or preventing the control processes from proper execution.

Our approach does not mitigate network attacks such as man-in-the-middle or DoS attacks that restrict network access. Also, the safety guarantees of the physical plant do not hold if the system is susceptible to sensor jamming attacks (\eg GPS jamming, electromagnetic interference on sensors \etc). 

%
%
%
%

\section{Background on Safety Controller}
\label{sec:background}

As we explain in Section \ref{sec:details}, a core component of our approach is the \textit{safety controller~(SC)}. Since the properties of the SC is essential for our approach, we now introduce this concept and illustrate a method to construct the SC. In the next section, we will use these concepts to establish the restart-based protection. For ease of understanding we first present some useful definitions before describing the details of SC design.

\begin{definition}[Admissible and Inadmissible States]
States that do not violate any of the operational constraints are referred to as \textit{admissible states} and denoted by $\mathcal{S}$. Likewise those that violate the constraints are referred to as \textit{inadmissible states} and denoted by $\mathcal{S}'$.
\end{definition}



\begin{definition}[Recoverable States]
\textit{Recoverable states}, a subset of the admissible states, such that if any given SC starts operations from one of those states, all future states will remain admissible. The set of recoverable states is denoted by $\mathcal{R}$.
\end{definition}



To implement our restart-based protection approach, one must construct an SC for the system and find its associated $\mathcal{R}$. For this, we provide a brief overview of constructing a SC introduced in earlier literature~\cite{seto1999case}\footnote{Any other design method for constructing an SC and the corresponding recoverable region~$\mathcal{R}$ with the above properties can be utilized for the purpose of the framework introduced in this paper.}. 
According to this approach, 
the SC is designed by approximating the system with linear dynamics in the form of $\dot{x} = Ax + Bu$, for state vector $x$ and input vector $u$. In addition
\emph{the safety constraints of the physical system are expressed as linear
constraints} in the form of $H \cdot x \leq h$ where $H$ and $h$ are constant matrix and vector. Consequently, the set of admissible states are $\mathcal{S} = \{x: H \cdot x \leq h \}$.

Safety constraints, along with the linear dynamics for the system are the inputs to a convex optimization problem. 
These parameters
produce both linear proportional controller gains $K$ as well as a positive-definite matrix $P$. The resulting linear-state feedback controller, $u = Kx$, yields closed-loop dynamics in the form of $\dot{x} = (A + BK)x$. Given a state $x$, when the input $u = Kx$ is used, the $P$ matrix defines a Lyapunov potential function $(x^TPx)$ with a negative-definite derivative. As a result, for the states where $x^TPx < 1$, the stability of the linear system is guaranteed using Lyapunov's
direct or indirect methods.
It follows that the states which satisfy $x^TPx < 1$ are a subset of the safety region. As long as the system's state is inside $x^TPx < 1$ and the $u = Kx$ is the controller, the physical plant will be driven toward the equilibrium point, \ie $x^TPx = 0$. Since the potential function is strictly decreasing over time, any trajectory starting inside the region $x^TPx < 1$ will remain there for an unbounded time window. 
As a result
no inadmissible states will be reached. 
Hence, the 
linear-state feedback controller $u = Kx$ is the SC and $\mathcal{R} = \{ x : x^TPx < 1\}$ is the recoverable region. Designing SC in such way ensures that the physical system would remain always safe \cite{Sha01usingsimplicity}.

\section{Restart-Based Protection}
\label{sec:details}


The goal of our proposed restart-based security approach is to \textit{prevent adversaries from damaging the physical plant}. The high-level overview of the proposed approach is illustrated in Fig.~\ref{fig:eventorder}.
The key idea is that, after a reboot (while the software is still uncompromised) the time instant to initiate the next restart must be decided such that the following conditions remain satisfied:

\begin{itemize}
	\item From the current time instant until the completion of the next reboot state will remain inside the admissible region even if the adversary is in full control and  
	\item  The state of the plant after completion of the next reboot will be such that the SC can stabilize the system. 
\end{itemize}

\begin{figure}[!tb]
	\begin{center}
		\includegraphics[width=0.40\textwidth]{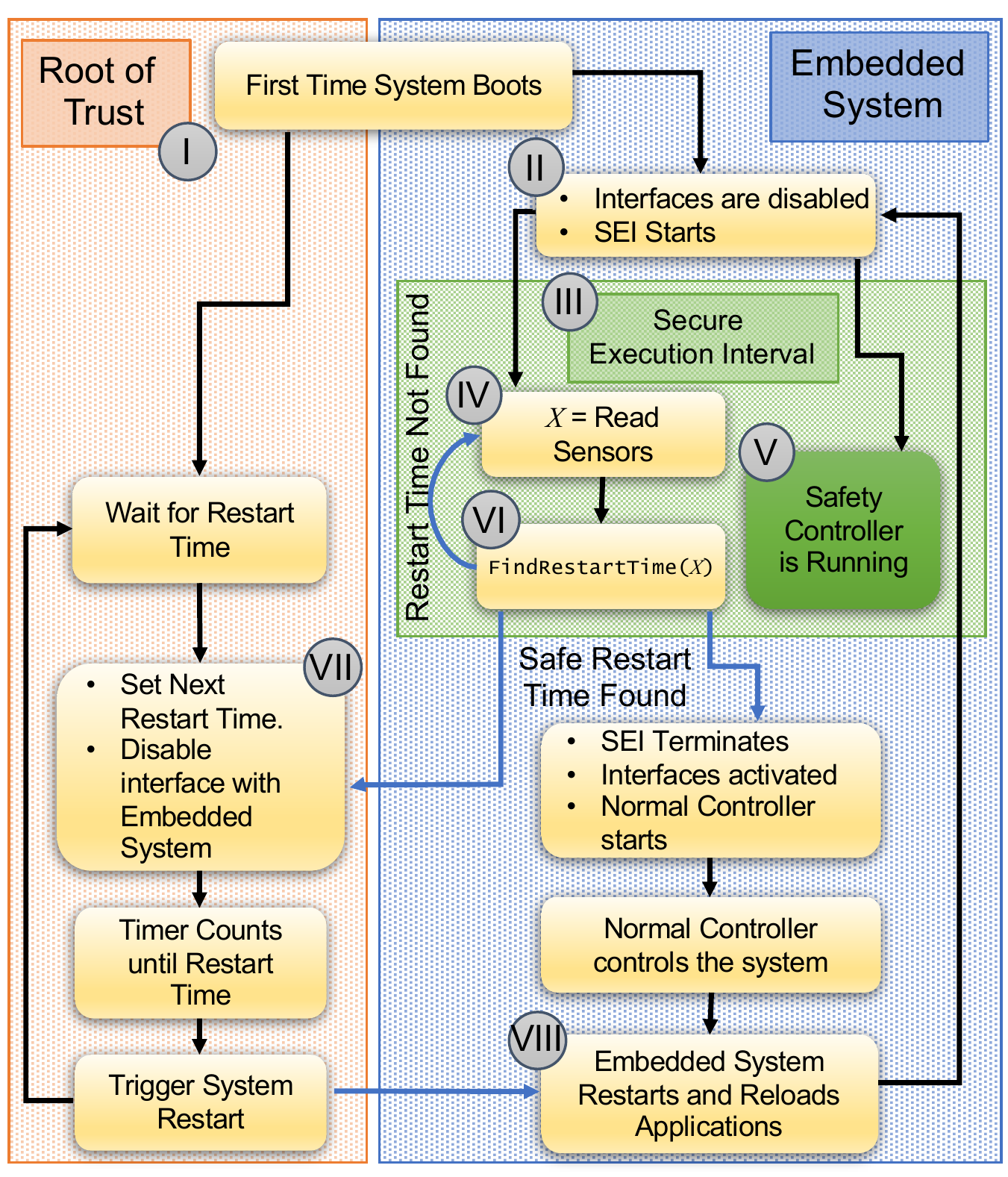}
		\caption{Sequence of the events.}
		\label{fig:eventorder}
	\end{center}
\end{figure}

After every restart (Marker VIII in Fig.~\ref{fig:eventorder}), the goal is to set the time of the next restart (Marker VII) such that it satisfies the above conditions (Marker VI). However, if the physical plant's state is too close to the boundary between the admissible and inadmissible states, there may not exist a time value that satisfies these conditions -- since, in such states, the system has no tolerance for unexpected behavior and if the adversary takes over, it can damage the plant in a very short time.
In such cases, the SC remains in charge of the control (illustrated by Marker V) for longer time and pushes the plant further into the set of admissible states $\mathcal{S}$ (refer to Section~\ref{sec:background}) until the safety margin is large enough and the restart time can be calculated~(system moves back and forth between states IV and VI until a restart time is available).

Three challenges need to be addressed for this approach to be practical. First, we need a \textit{mechanism to compute the next restart time after each reboot}. Second, we need to \textit{ensure the integrity of the process that performs this computation}. Third, once the next restart time is computed, \textit{a mechanism is needed to enforce its execution at designated time despite full adversarial control of the embedded system}. 

To address the first challenge, an approach for computing the safe restart time (Marker V in Fig.~\ref{fig:eventorder}) is presented in Section~\ref{sec:securetime}. To guarantee integrity of restart time computation~(second challenge), we introduce the idea of \textit{secure execution intervals~(SEI)} as illustrated by Marker III in Fig~\ref{fig:eventorder}. After every restart, for a short interval of time, all the external interfaces of the system remain disabled (Marker II). This allows the reliable execution of SC and for trusted computation of the restart time without adversarial interference. SEI is discussed in more detail in Section~\ref{sec:SEI}. Furthermore, to ensure that the intruder cannot prevent the system from restarting~(third challenge), the design includes a simple, isolated hardware module (illustrated by Marker I) called~\emph{Root of Trust~(RoT)} discussed in Section~\ref{sec:RoTDesign}.

It is useful to define some parameters before the detailed discussion on the components. In the rest of this section, $T_s$ is the length of the SEI and is fixed throughout operation, and $t$ is the current time. In addition, $T_r$ is the time it takes for the embedded system to restart and reboot.

\subsection{Secure Execution Interval~(SEI)}
\label{sec:SEI}

The need for SEIs arises from the fact that a mechanism is required to prevent any malicious interference towards the tasks that are necessary for providing the safety guarantees. After a restart, while all the external interfaces of the system remain disabled, the system software is loaded into memory from a read-only storage unit. Disabling interfaces isolates the system and protects it from intruders. So long as the system is isolated, we assume that the software is uncorrupted. SEI refers to this isolated interval of time after each restart. 

During the SEI, two tasks execute in parallel: SC (which runs periodically to keep the physical system stable) and \texttt{FindRestartTime} (Section~\ref{sec:securetime}). If \texttt{FindRestartTime} cannot find a safe restart time~(this may occur when the physical plant is very close to being in the  inadmissible region), SEI is extended for another $T_s$ time units\footnote{As we illustrate in Section~\ref{sec:evaluation}, this extension does not negatively affect the safety.}. Extending SEI gives SC more time to push the state further away from the unsafe boundary and into the admissible region. If \texttt{FindRestartTime} task is able to find a restart time, \ca the \texttt{SetRestartTime} interface of the RoT~(details in next section) is used to set the time for the next restart event, \cb SEI terminates, \cc all the interfaces are activated, \cd the main controller of the system is activated and \ce the  system enters normal operation mode -- until the next restart takes place. Note that following this procedure, the \texttt{SetRestartTime} interface of RoT will be called exactly \emph{once} before the SEI terminates.

\subsection{Hardware Root of Trust~(RoT)}
\label{sec:RoTDesign}

Our design requires a mechanism to ensure that under any circumstances, the adversary cannot prevent the system from restarting. Hence, we include an isolated HW module in charge of triggering restarts and refer to it as hardware root of trust~(RoT).

RoT provides a hardware interface, $\texttt{SetRestartTime}$, that during the SEI allows the processor to set the time of the next restart event.  To achieve this, after each restart, RoT allows  the processor to call the \texttt{SetRestartTime} interface only \emph{once}. Additional calls to this interface will be ignored. RoT immediately sets a timer for the time value received from \texttt{SetRestartTime} interface and issues a restart signal to the hardware restart pin of the platform upon its expiration.


In order to achieve the security goals of the platform, ROT needs to be secure and incorruptible. Hence, we require hardware isolation~(\eg a standalone timer) and independence from the processor with no connectivity except for the HW interface \texttt{setRestartTime}. In our prototype implementation, RoT is implemented using a simple micro-controller (Section~\ref{sec:evaluation}).

%
%

\subsection{Finding a Safe Restart Time}
\label{sec:securetime}



In this section, we now discuss the implementation of the \texttt{FindRestartTime} task. This task is activated immediately after system reboot, along with the SC task  and calculates the time for the next restart event.

Before we proceed, it is useful to define some notations. We use the notation of $\mathsf{Reach}_{=T}(x,C)$ to denote the set of states that are reachable by the physical plant from an initial set of states $x$ after exactly $T$ units of time have elapsed under the control of controller $C$~\cite{rt-reach}. $\mathsf{Reach}_{\leq T}(x,C)$ can be defined as $\bigcup_{t=0}^{T} \mathsf{Reach}_{=t}(x,C)$ \ie set of all the states reachable within up to $T$ time units. In addition, we use $SC$ to refer to the \emph{safety controller} and $UC$ to refer to an \emph{untrusted controller}, \ie one that is compromised by the adversary.

Conditions presented in our newly developed Theorem~\ref{thm:safety} evaluates whether the system remains safe and recoverable if restarted after $\delta_r$ time units.

\begin{theorem}
Assuming that $x(t)$ is the state of the system at time $t$ and the adversary controls the system~(UC) from $t$ until the restart of the system. Under the following conditions, the embedded system can be restarted at time $t + \delta_r$ thus guaranteeing that attacker cannot violate the safety requirements of the physical plant:
\begin{enumerate}
	\item $\mathsf{Reach}_{\leq \delta_r + T_r}(x, UC) \subseteq
	\mathcal{S}$\label{cond:hardconst_reach_pre_tc};
	\item $\mathsf{Reach}_{\leq T_\alpha}(\mathsf{Reach}_{= \delta_r + T_r}(x, UC), SC)
	\subseteq \mathcal{S}$\label{cond:hardconst_reach_pre_ts};
	\item$\mathsf{Reach}_{= T_\alpha }(\mathsf{Reach}_{= \delta_r+T_r}(x, UC), SC)
	\subseteq \mathcal{R}$\label{cond:hardconst_reach_ts}.
\end{enumerate}
\label{thm:safety}
\end{theorem}

\begin{proof}
Intuitively, this condition says that plant remains safe if, \ca the adversary~(UC) cannot reach an inadmissible state before the next restart completes ($\delta_r + T_r$), \cb if the safety controller takes over after next restart, it will avoid unsafe states until $\delta_r + T_r + T_\alpha$ time units passes and \cc after $\delta_r + T_r + T_\alpha$ time, a state in $\mathcal{R}$ will be safely reached.

To prove formally, assume by contradiction that the system is not safe under these conditions. Hence, an inadmissible state is reached at some time. This time will be either less than $\delta_r + T_r$, more than $\delta_r + T_r$ and less than $\delta_r + T_r + T_\alpha$ or more than $\delta_r + T_r + T_\alpha$. The first two of these cases are ruled out directly by conditions (1) and (2), so only the third case needs to be examined.

From properties of SC, $\mathsf{Reach}_{\leq \infty}(\mathcal{R}, SC) \cap \mathcal{S}' = \emptyset$. Since if $\mathcal{R}' \subseteq \mathcal{R}$, $\mathsf{Reach}_{\leq \infty}(\mathcal{R}', SC) \subseteq \mathsf{Reach}_{\leq \infty}(\mathcal{R}, SC)$, the smaller set of states $\mathcal{R}' = \mathsf{Reach}_{= T_\alpha }(\mathsf{Reach}_{= \delta_r+T_r}(x, UC), SC) \subseteq \mathcal{R}$ will also satisfy the condition $\mathsf{Reach}_{\leq \infty}(\mathcal{R}', SC) \cap \mathcal{S}' = \emptyset$. Therefore, every state reached after $\delta_r + T_r + T_\alpha$ will be admissible\footnote{This proof is adopted from work~\cite{rt-reach}.}.
\end{proof}

As a direct result of this theorem, after each restart the safety controller remains in charge until a $\delta_r$ is found that satisfies the above conditions. Algorithm~\ref{alg:find_restart_time} shows the pseudo-code of system operation. 
Lines~\ref{line:startoffind} to~\ref{line:endoffind} represent the $\texttt{FindRestartTime}(x)$ function which uses the conditions of Theorem~\ref{thm:safety} to find a safe restart time.

\begin{algorithm}[ht]
	
	\newcommand{\algorithmicbreak}{\textbf{break}}
	\newcommand{\BREAK}{\STATE \algorithmicbreak}
	\renewcommand\algorithmiccomment[1]{%
		{\it /* {#1} */} %
	}
	\renewcommand{\algorithmicrequire}{\textbf{Input:}}
	\renewcommand{\algorithmicensure}{\textbf{Output:}}
	
	\begin{algorithmic}[1]
		\begin{footnotesize}
			\STATE Start SC. ~~\COMMENT{SEI begins}\label{line:begin}
			\STATE 
            $\mathtt{timeFound} := $  \textbf{False} 
			\WHILE{$\mathtt{timeFound} ==$ \textbf{False}}
			
			\STATE $x:=$ most recent state of the system from Sensors \label{line:startoffind}

			\STATE 
  $\mathtt{\delta_{\text{candidate}}} := \delta_{\text{init}}$ ~~\COMMENT{initialize the restart time}

			\STATE $\mathtt{startTime}$ = currentTime()
			\WHILE{currentTime() - $\mathtt{startTime}$ \textless $T_s$}
			\IF {conditions of Theorem~\ref{thm:safety} are true for $\delta_{\text{candidate}}$} \label{line:thm}
			\STATE $\delta_{\text{safe}} := \delta_{\text{candidate}}$
			\STATE $\delta_{\text{candidate}}:= \delta_{\text{candidate}} + \mathsf{INC\_STEP}$ ~\COMMENT{increase the $\delta_{\text{candidate}}$}
			\ELSE 
			\STATE $\delta_{\text{candidate}}:= \delta_{\text{candidate}} - \mathsf{INC\_STEP}$
			~\COMMENT{decrease the $\delta_{\text{candidate}}$}
			\ENDIF
			\ENDWHILE		\label{line:endoffind}	
			
			\IF{$\delta_{\text{safe}} >$ currentTime() - startTime} \label{line:elapsedTimeStart}
			\STATE$\delta_{\text{safe}} = \delta_{\text{safe}} - $(startTime = currentTime())			
			\STATE $\mathtt{timeFound} := $ \textbf{True}			
			\ENDIF \label{line:elapsedTimeEnd}

			\ENDWHILE
			\STATE Send $\delta_{\text{safe}}$ to RoT. ~~\COMMENT{Set the next restart time.}	
			
			\STATE Activate external interfaces.~~\COMMENT{SEI ends.}
			\STATE Terminate SC and start the main controller.
			
			\STATE When RoT sends the restart signal to hardware restart pin:
            \STATE \hspace*{1em} Restart the system 
            \STATE \hspace*{1em} Repeats the procedure from beginning (\eg from Line~\ref{line:begin})
			
			
		\end{footnotesize}
		
	\end{algorithmic}
	\caption{Pseudo-code of system operation}
	\label{alg:find_restart_time}
\end{algorithm}

%
%
%
%
%
%
%

Note that the computation time of each round of \texttt{FindRestartTime} task (Lines~\ref{line:startoffind} to~\ref{line:endoffind}) is capped by $T_s$ so that it samples the state frequently enough. In addition, evaluation of the conditions of Theorem~\ref{thm:safety} in Line~\ref{line:thm} requires finite time and is not instantaneous. To adjust for it, the elapsed time is deducted from the computed time~(Lines~\ref{line:elapsedTimeStart} to~\ref{line:elapsedTimeEnd}) when the restart time is being set in RoT.

It is worth mentioning that in this work, the run-time computation of reachable states~(\ie $\mathsf{Reach}$ function) that is used in evaluation of the Theorem~\ref{thm:safety} conditions, is performed using the real-time reachability technique that is proposed earlier~\cite{rt-reach}. This approach uses the model of the system dynamics~(linear or non-linear) to compute reachability. Since the real actions of the adversary at run-time are unknown reachability of the system under \emph{all} possible control values is calculated (compute reachability under $UC$). As a result, the reachable set under $UC$ is the largest set of states that might be reached from the given initial state, \textit{within the specified time}.

\section{Security Improvement By Restarting on Fixed Periods}
\label{sec:probability}

The approach discussed thus far, aims to achieve a very strong goal for the safety and security of CPS. However, trying to reach such strong protection guarantees for systems with fast-moving dynamics or narrow security margin~(\eg UAVs and power grid components), may result in very short restart times. Because, such systems can quickly become unstable. With such short restart times, actions such as establishing Internet-based communication, authenticating with a remote server, or performing computationally heavy tasks may become infeasible. In this section, our goal is to demonstrate that restarting the system, even when the restart time is not strictly calculated using the approach of Section~\ref{sec:securetime}, may improve the security for many systems. 

Assume $\mathcal{P}(t)$ is the \emph{probability of successful damage} that answers to the following question: "What is the probability that an attacker succeeds in damaging the system within certain $t$ time units?". In this context, \emph{damaging the system} refer to the act of forcing the plant to reach a state in the inadmissible region. Here, $t=0$ is the start of system operation from a functional state. The \emph{expected time of damage} is $E^{\mathcal{P}}(t_{\text{attack}}) = \int_0^\infty \tau \mathcal{P}(\tau).d\tau$. A countermeasure is assumed to \emph{improve the security} of the system if $E^{\hat{\mathcal{P}}}(t_{\text{attack}}) < E^{\mathcal{P}}(t_{\text{attack}})$ where $\hat{\mathcal{P}}(t)$ and $\mathcal{P}(t)$ are the probability of successful damage with and without that countermeasure in place. In the rest of this section, we will provide an intuitive discussion of the types of systems for which periodic restarting of the system improves the security \ie reduces the expected time of attack.

To proceed with the discussion, we assume that $\mathcal{F}(t)$ is the probability density function~(PDF) of the successful attack on the system whose value at any given time represents the relative likelihood of an attacker being able to cause damage to the physical system. The relationship between the probability of a successful attack and its PDF is  $\mathcal{P}(t) = \int_0^{t}\mathcal{F}(\tau)\cdot d\tau$.  Further, let us assume that $\hat{\mathcal{F}}(t)$ is the PDF of a successful attack on the same system when it is restarted periodically every $\delta_r$ times.
We have the following: 
\begin{equation}
\hat{\mathcal{F}}(t) = (1-\hat{\mathcal{P}}(k\delta_r))\mathcal{F}(t-k\delta_r)
\label{eq:resetPDF}
\end{equation}
where $k = \lfloor t/\delta_r \rfloor$. The term $(1-\hat{\mathcal{P}}(k\delta_r))$ is the probability of an attack not taking place in the previous $k$ restart cycles and, $\mathcal{F}(t-k\delta_r)$ is the first $\delta_r$ time units of function $\mathcal{F}$. Integrating over Equation~\ref{eq:resetPDF} results in Equation~\ref{eq:resetProb} for the probability of successful attacks with the restart-based protection.
\begin{equation}
\hat{\mathcal{P}}(t) = \hat{\mathcal{P}}(k\delta_r) + (1-\hat{\mathcal{P}}(k\delta_r))\mathcal{P}(t-k\delta_r)
\label{eq:resetProb}
\end{equation}

From equations~\ref{eq:resetPDF} and~\ref{eq:resetProb}, we can see that  $\hat{\mathcal{F}}$ and $\hat{\mathcal{P}}(t)$  are only a function of the values of $\mathcal{F}$~( and consequently $\mathcal{P}$) within the first $\delta_r$ time units.  The effectiveness of the restart-based approach on reducing probability of damage $\hat{\mathcal{P}}$ depends on two main factors: (1) distribution of probability density of the $\mathcal{F}$ within the starting $\delta_r$ segment. In other words, smaller density of $\mathcal{F}$ within the first $\delta_r$ time units leads to a smaller growth rate of $\hat{\mathcal{P}}$ which eventually results in a smaller $E^{\hat{\mathcal{P}}}(t_{\text{attack}})$. And, (2) the second factor is the growth rate of the original $\mathcal{P}$ in the time $t > \delta_r$~(\ie how does the probability of the attack without restarts is divided over time). 

We further clarify this analysis with the use of four illustrative systems in the following. Figure~\ref{fig:fig11} depicts the PDF~($\mathcal{F}$, blue line) and probability of damage function~($\mathcal{P}$, yellow line) for a demonstrative system. It also depicts the new damage PDF~($\hat{\mathcal{F}}$, red line) and probability of damage function~($\hat{\mathcal{P}}$, purple line) after the periodic restarts with a period of $20$ time units is performed on the system. Notice that, restarting reduces the probability of damage over the time for this system~($\hat{\mathcal{P}} < \mathcal{P}$). Figure~\ref{fig:fig12} depicts the same functions for a different system where the original PDF has a smaller density within the first $20$ time units~($\mu$ of the normal function in Figure~\ref{fig:fig12} is larger than Figure~\ref{fig:fig11}). As a result of this difference, $\hat{\mathcal{F}}$ is overall smaller in Figure~\ref{fig:fig12} which reduces the growth rate of $\hat{\mathcal{P}}$ in Figure~\ref{fig:fig12} compared to~\ref{fig:fig11}. This example illustrates that periodic restart-based protection is more effective on the systems where the probability of attack/damage is low within the first $\delta_r$ times of the operation.

\begin{figure}[h]
	\begin{center}
		\centering 
		\subfigure[$\mathcal{F}$ is a normal function with $\mu = 30$ time units $\delta = 40$. ]{\label{fig:fig11}\includegraphics[width=0.45\textwidth]{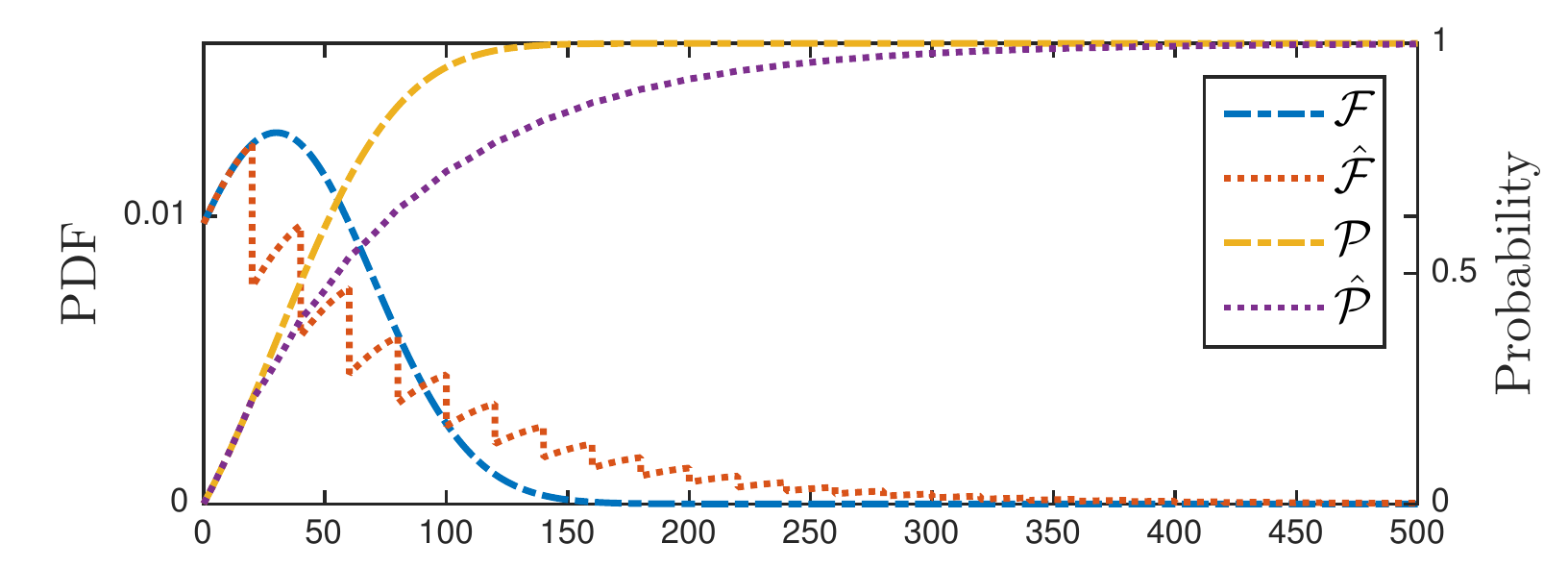}}
		\subfigure[$\mathcal{F}$ is a normal function with $\mu = 80$ time units and $\delta = 40$.]		{\label{fig:fig12}\includegraphics[width=0.45\textwidth]{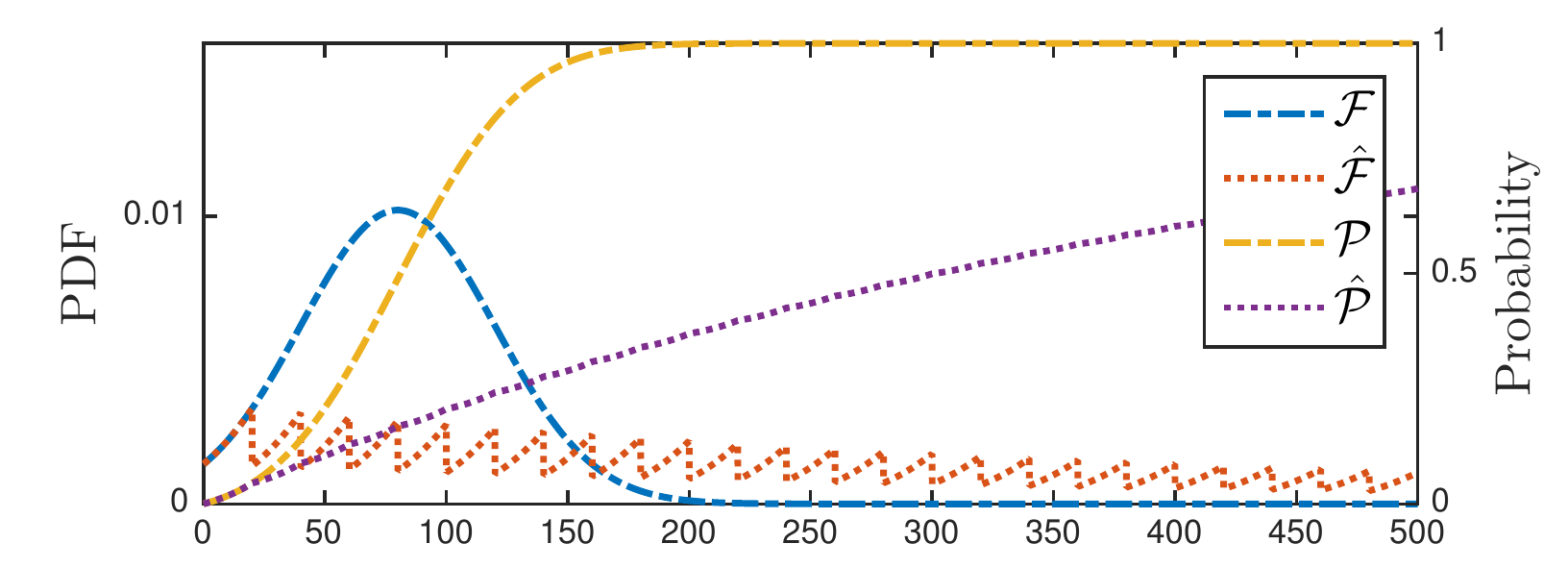}}
		\caption{In both figures, restart time is $20$ time units. Restart-based protection results in a smaller $\hat{\mathcal{P}}$ in the system of Figure~\ref{fig:fig12}~(the purple line) because of the smaller density of function $\mathcal{F}$ within the first $20$ time units.}
		\label{fig:resetPDF1}
	\end{center}
\end{figure}

Figures~\ref{fig:fig21} and~\ref{fig:fig22} illustrate the $\mathcal{F}$,$\mathcal{P}$, $\hat{\mathcal{F}}$, and $\hat{\mathcal{P}}$ functions for two systems where $\mathcal{F}$ has the exact same values over the first $20$ time units in both systems. Notice that, the damage probability function with restarting~($\hat{\mathcal{P}}$) is the same in both systems. However, for the system of Figure~\ref{fig:fig21}, the second major attack probability occurs later than the system of Figure~\ref{fig:fig22}. As seen in \ref{fig:fig21}, function $\hat{\mathcal{P}}$ has larger values than $\mathcal{P}$ in the range of 100 to 380 time units. This example illustrates a system where periodic restart-based approach does not improve the security of the system. 

\begin{figure}[h]
	\begin{center}
		\centering 
		\subfigure[$\mathcal{F}$ consists of two normal functions with mean at  $30s$ and $350s$ and a standard deviation of 40. ]{\label{fig:fig21}\includegraphics[width=0.45\textwidth]{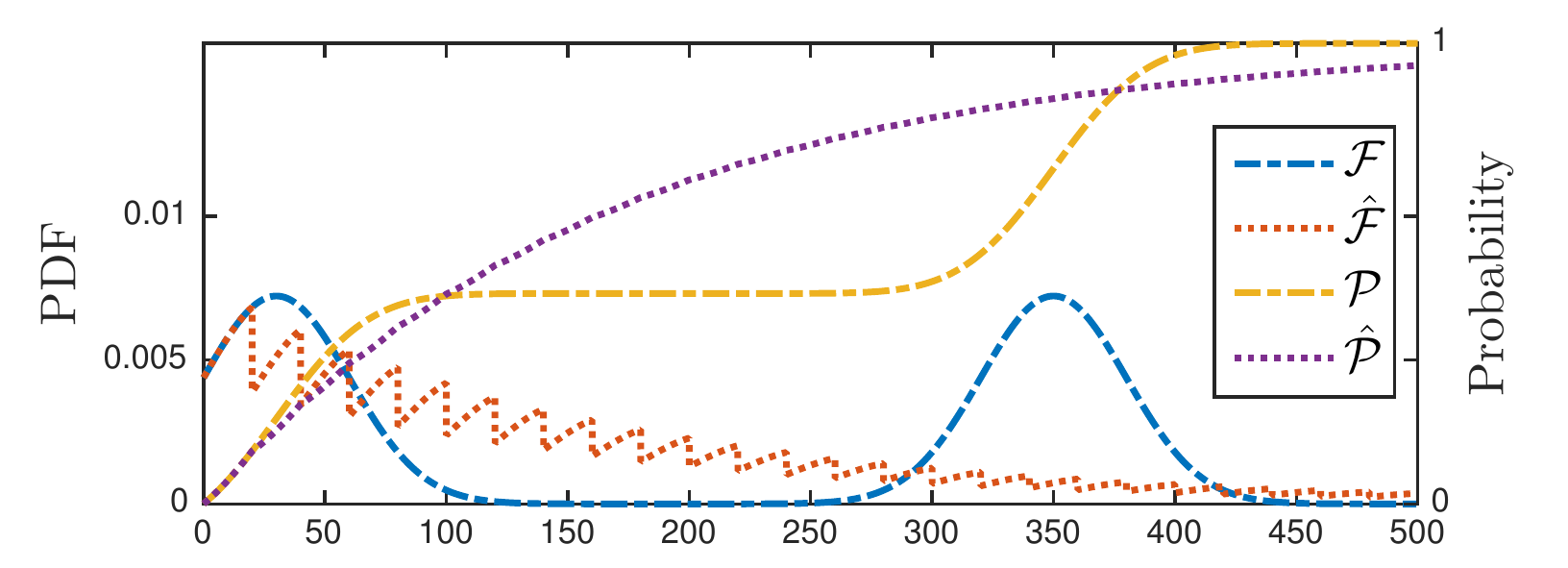}}
		\subfigure[$\mathcal{F}$ consists of two normal functions with mean at  $30s$ and $150s$ and a standard deviation of 40.]		{\label{fig:fig22}\includegraphics[width=0.45\textwidth]{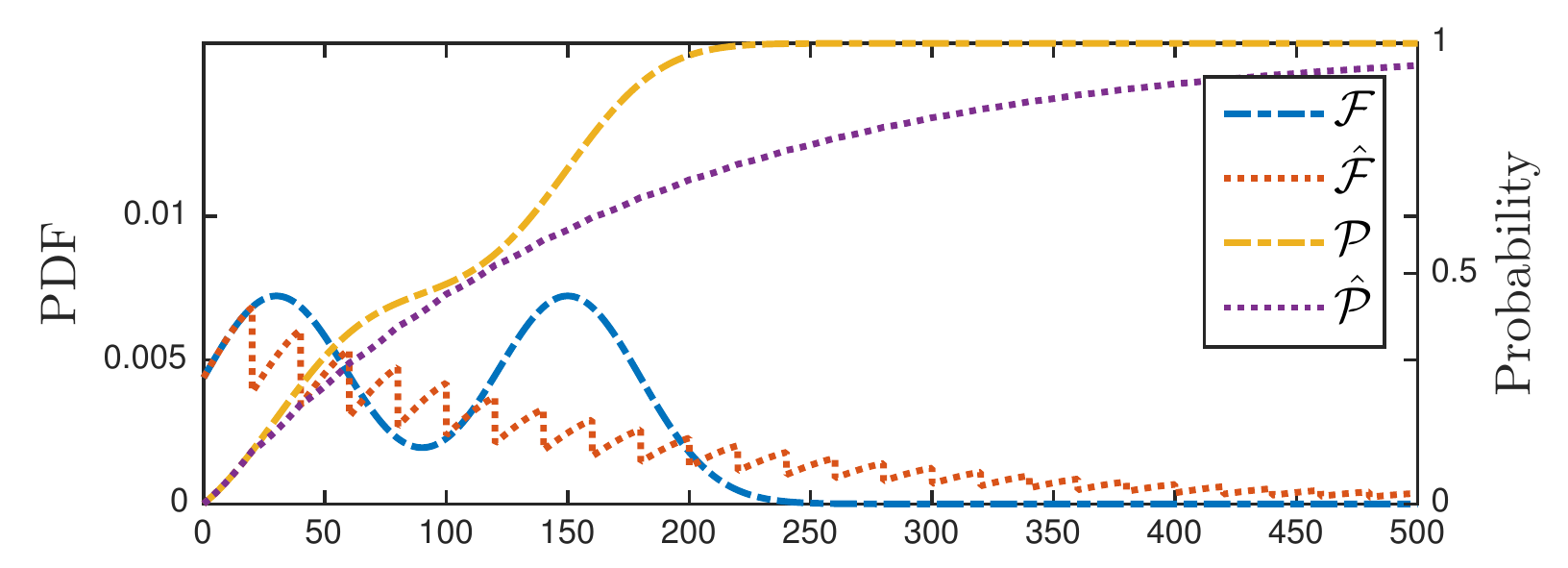}}
		\caption{In both figures, restart time is $20s$ and the density of $\mathcal{F}$ within the first 20 seconds is very similar. Restart based protection is more effective in reducing $\hat{\mathcal{P}}$ in the system of Figure~\ref{fig:fig22}~(the purple line).}
		\label{fig:resetPDF2}
	\end{center}
\end{figure}

In summary, equations~\ref{eq:resetPDF} and~\ref{eq:resetProb} provide an analytical framework for system designers to utilize $\mathcal{F}$ and $\mathcal{P}$ functions~(that are properties of the particular system) and evaluate the effectiveness of restart-based protection for the target systems.  Notice that, there is an extensive body of work in the literature that studied  probabilistic risk assessment of security threats for various domains and platforms~\eg \cite{schneier1999attack,Arnold2014,kumar2015time,5399279} . Those results can be utilized to construct $\mathcal{F}$ and $\mathcal{P}$ functions for the target systems.

\section{Evaluation}
\label{sec:evaluation}

We now evaluate various aspects of restart-based protection approach using a combination of simulation-based and an actual implementation. To evaluate practicality, we compare the safe restart times for two target systems, a 3-degree of freedom helicopter~\cite{3dof} and a warehouse temperature system~\cite{trapnes2013optimal}. The former is a very unstable system with fast-moving dynamics and the latter is a relatively stable system with slower dynamics. In addition, we have implemented a controller on a commercial-off-the-shelf~(COTS) ZedBoard~\cite{zboarddatasheet} embedded platform for a realistic 3DOF helicopter. We perform platform attacks on the embedded system and empirically demonstrate that the helicopter remains safe and recovers from the attacks.

\subsection{Physical Systems Under Evaluation}

\subsubsection{3-Degree of Freedom Helicopter:} 3DOF helicopter~(displayed in figure~\ref{fig:3dof}) is a simplified helicopter model, ideally suited to test intermediate to advanced control concepts and theories relevant to real-world applications of flight dynamics and control in the tandem rotor helicopters, or any device with similar dynamics~\cite{3dofhelicopter}. It is equipped with two motors that can generate force in the upward and downward direction, according to the given actuation voltage. It also has three sensors to measure elevation, pitch and travel angle as shown in Figure~\ref{fig:3dof}. We use the linear model of this system obtained from the manufacturer manual~\cite{3dofhelicopter} for constructing the safety controller and calculating the reachable set in run-time. Due to the lack of space, the details of the model are included in our technical report~\cite{githubrepoemsoft}.

For the 3DOF helicopter, the safety region is defined in such a way that the helicopter fans do not hit the surface underneath, as shown in Figure~\ref{fig:3dof}. The linear inequalities describing the safety region are provided are $-\epsilon + |\rho| \leq 0.3$ and $|\rho| \leq \pi/4$. Here, variables $\epsilon$, $\rho$, and $\lambda$ are the elevation, pitch, and travel angles of the helicopter. Limitations on the motor voltages of the helicopter are $|v_l| \leq 1.1$ and $|v_r| \leq 1.1$ where $v_l$ and $v_r$ are the voltage for controlling left and right motors.

%

\subsubsection{Warehouse Temperature System:}
This system consists of a warehouse room 
with a direct conditioner system~(heater and cooler) to the room and another conditioner in the floor. The safety goal for this system is to keep the temperature of the room within the range of $[20^{\circ}C,30^{\circ}C]$. This system and its model are obtained from the work in~\cite{trapnes2013optimal}. Equations~\ref{eq:heatermodel1} and~\ref{eq:heatermodel2} describe the heat transfer between the heater and floor, floor and the room, and room and outside temperature. The model assumes constant mass and volume of air and heat transfer only through conduction.

\begin{equation}
\dot{T}_F = -\frac{U_{F/R} A_{F/R}}{m_FCp_F}(T_F-T_R) + \frac{u_{H/F}}{m_FCp_F}
\label{eq:heatermodel1}
\end{equation}
\begin{equation}
\dot{T}_R = -\frac{U_{R/O} A_{R/O}}{m_RCp_R}(T_R-T_O) + \frac{U_{F/R}A_{F/R}}{m_RCp_R}(T_F-T_R) + \frac{u_{H/R}}{m_RCp_R}
\label{eq:heatermodel2}
\end{equation}
Here, $T_F$, $T_R$, and $T_O$ are temperature of the floor, room and outside. $m_F$ and $m_R$ are the mass of floor and the air in the room. $u_{H/F}$ is the heat from the floor heater to the floor and $u_{H/R}$ is the heat from the room heater to the room both of which are controlled by the controller. $Cp_F$ and $Cp_R$ are the specific heat capacity of floor~(in this case concrete) and air. $U_{F/R}$ and $U_{R/O}$ represent the overall heat transfer coefficient between the floor and room, and room and outside.

For this experiment, the walls are assumed to consist of three layers; the inner and outer wall are made of oak, and isolated with rock wool in the middle. The floor is assumed to be quadratic and consists of wood and concrete. The parameters used are as following\footnote{For the details of calculation of $U_{F/R}$ and $U_{R/O}$ and the values of the parameters refer to Chapter 2 and 3 of~\cite{trapnes2013optimal}.}: $U_{R/O} = 539.61J/hm^2K$, $U_{F/R} = 49920 J/hm^2K$ , $m_R = 69.96kg$, $m_F = 6000kg$, floor area $A_{F/R} = 25 m^2$, wall and ceiling area $A_{R/O} = 48 m^2$, thickness of rock wool, oak and concrete in the wall and floor respectively $0.25m$, $0.15m$ and $0.1m$. Maximum heat generation capacity of the room and floor conditioner is respectively $800J/s$ and $115J/s$. And, maximum cooling capacity of the room and the floor cooler is $-800J/s$ and $-115J/s$.

\subsection{Safe Restart Time for the Physical Systems}

\label{sec:restarttimeeval}


As discussed in the previous sections, after each system reboot the time of the next restart of the embedded system needs to be computed. The main factor that impacts the safe restart time is the proximity of the current state of the plant to the boundaries of the inadmissible states. In this section, we demonstrate this point on two physical systems with fast and slow dynamics; 3DOF helicopter and warehouse temperature system.


In figures~\ref{fig:exp1} and~\ref{fig:exp2}, the calculated safe restart times are plotted for the two systems under investigation. In these figures, the red region represents the inadmissible states and the plant must never reach to those states. If the plant is in a state that is marked green, it is still undamaged. However, at some future time it will reach an inadmissible state and the safety controller may not be able to prevent it. This is because physical actuators have a limited range and the maximum capacity of the actuators is not enough to cancel the momentum and prevent the plant from reaching the unsafe states. And the gray/black area is the region in which a value for the safe restart time of the system can be computed. In this region, the darkness of the color indicates the value of calculated restart time if the plant is in that state. The black points indicate the maximum time and the white points are an indicator of zero time.

\begin{figure}[ht]
	\begin{center}
		\centering 
		\subfigure[Projection of the state space into the plane $\dot{\epsilon} = 0$, $\dot{\rho} = 0$, $\lambda = 0$, and $\dot{\lambda} = 0.3\text{Radian}/s$]
			{\label{fig:exp11}\includegraphics[width=0.22\textwidth]{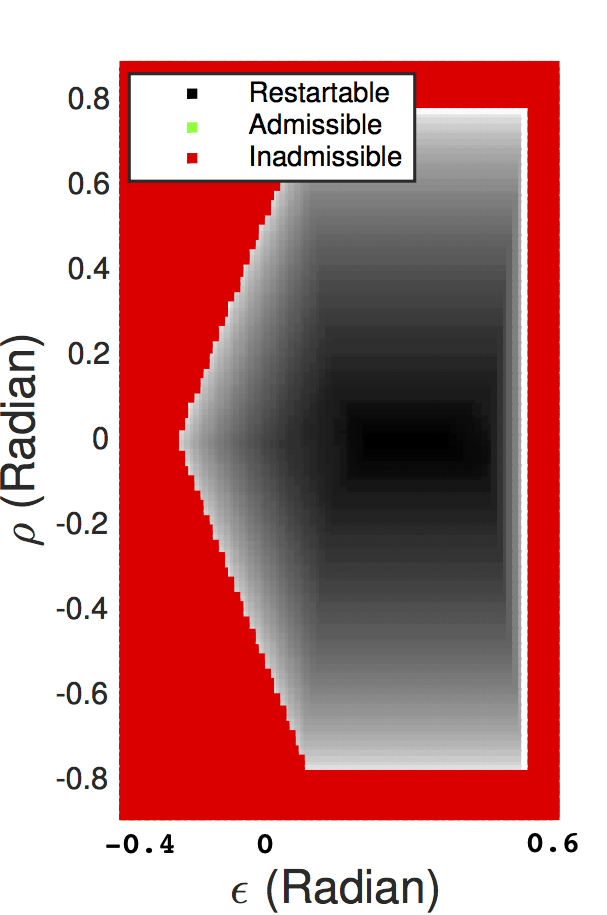}}
		\hfil
		\subfigure[Projection of the state space into the plane  $\dot{\epsilon} = -0.3\text{Radian}/s$, $\dot{\rho} = 0$, $\lambda = 0$, and $\dot{\lambda} = 0.3\text{Radian}/s$]
		{\label{fig:exp12}\includegraphics[width=0.22\textwidth]{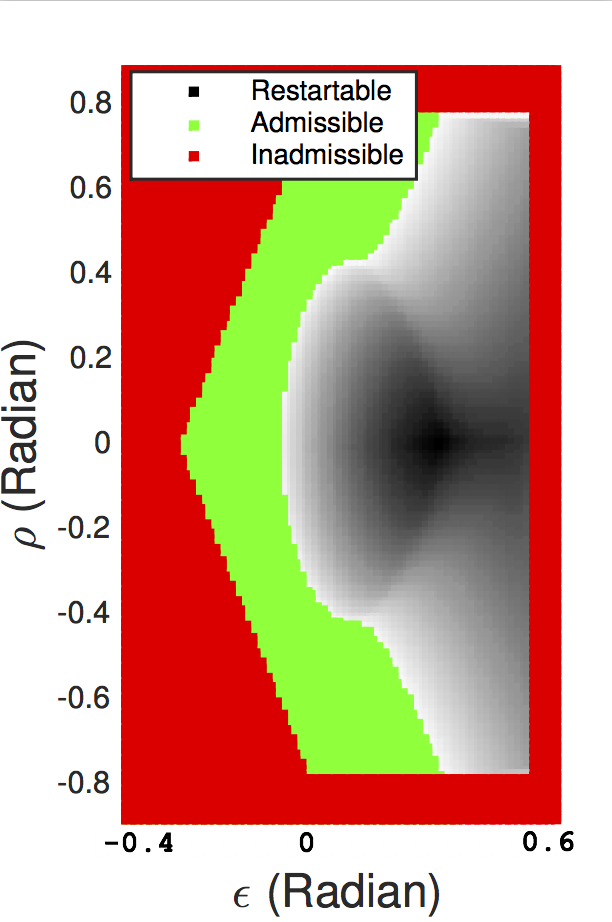}}
		\caption{Calculated safe restarting time for the 3DOF helicopter system from various states. Darkest black points represent a possible restart time of 1.23 seconds.}
		\label{fig:exp1}
	\end{center}
\end{figure}

In Figure~\ref{fig:exp1}, for 3DOF helicopter, the maximum calculated safe restart time~(\ie the darkest points in the gray region) is $1.23$ seconds. As seen in the Figure, restart time is maximum in the center where it is farthest away from the boundaries of unsafe states. In Figure~\ref{fig:exp12}, the angular elevation speed of the helicopter is $\dot{\epsilon} = -0.3\text{Radian}/s$. This indicates that the helicopter is heading towards the surface with the rate of 0.3 Radian per second. As a result, the plant cannot be stabilized from the lower elevations levels~(\ie the green region). Moreover, values of the safe restart times are smaller in Figure~\ref{fig:exp12} compared to the Figure~\ref{fig:exp11}. Because, crashing the 3DOF helicopter with the initial downward speed requires less time than the case without any downward angular velocity.

\begin{figure}[ht]
	\begin{center}
		\centering 
		\subfigure[Projection of states to $T_F = 25ºC$]
		{\label{fig:exp21}\includegraphics[width=0.22\textwidth]{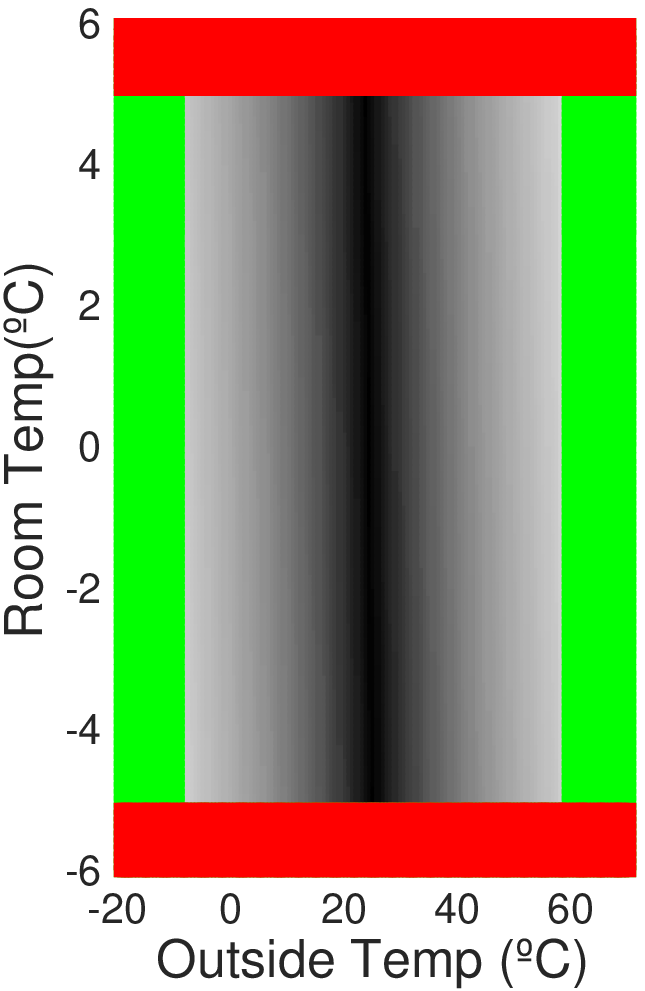}}
		\subfigure[Projection of states to $T_F = 29ºC$]
		{\label{fig:exp22}\includegraphics[width=0.22\textwidth]{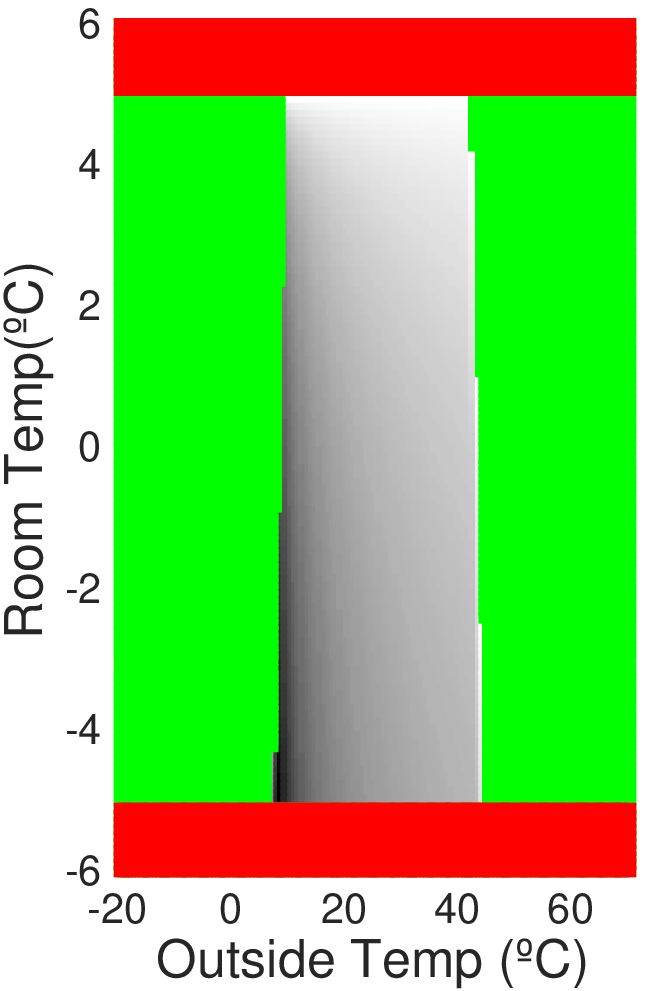}}
		\caption{Safe restarting times calculated for the embedded controller system of warehouse temperature.}
		\label{fig:exp2}
	\end{center}
\end{figure}

Figures~\ref{fig:exp21} and~\ref{fig:exp21}, plot the calculated safe restart times for the warehouse temperature system. For this system, when the outside temperature is too high or too low, the attacker requires less time to take the temperature beyond or bellow the safety range. One major difference of this system with the 3DOF helicopter is that due to the slower dynamics, the safe restart times have values up to $6235s$~(close to one hour and forty minutes). In the warehouse system, the rate of the change of the temperature even when the heater/coolers run at their maximum capacity is slow and hence, attacker needs more time to take the system to unsafe states.


If above systems are restarted at times smaller than the calculated safe restart windows, the physical plant is guaranteed to remain safe under an attacker who gains full control immediately after SEI ends. Achieving this goal is very reasonable for systems such as the warehouse temperature control and many other similar systems with applications in Internet of Things~(IoT), city infrastructure, and \etc. For systems with faster and more unstable dynamics such as the helicopter, the calculated restart times might become very short. The fast restart rate, even though theoretically possible, in practice may create implementation challenges. For instance, initializing a connection with the remote user over Internet may require more time than the calculated safe restart time. For our particular 3DOF controller implementation~(Section~\ref{sec:implementation}), we use serial port for communication which has a very small initialization overhead. Despite short restart times, we demonstrate that the system remains stable and functional.


\subsection{Restarting and  Embedded System Availability}
During the restart, embedded system is unavailable and actuators keep executing the last command that was received before the restart occurred. Results of Theorem~\ref{thm:safety} guarantees that such unavailability does not compromise safety and stability of the physical plant. However, low availability may slow down the progress of the system towards its control mission. In this section, we measure the average percent of the time that the embedded system is available when restart-based protection is implemented. The approach is as follows.

Availability of the controller is the ratio of the time that the main controller is running~(not rebooting and not in SEI) to the whole operation time of the system. In every restart cycle, availability is $ (\delta_r)/(\delta_r + T_s + T_r)$ where $\delta_r$ is the safe restart time~(main controller is active during $\delta_r$), $T_s$ is the length of SEI, and $T_r$ is the reboot time of the embedded platform. Value of $\delta_r$ is not fixed and changes according to the state of the plant. To measure availability, we divide the state space into three sub-regions based on how often the physical plant is expected to operate in those regions.  The availability values computed for the states in the most common region are assigned a weight of \emph{1}, second most common region a weight of \emph{0.5} and the least common region a weight of \emph{0.3}. We calculated the expected availability of the controller for all the states in the operational range of the system.  At the end, the weighted average of the availability was calculated according to the aforementioned weights. For the value of the $T_r$, we used $390ms$ for the reboot time of the 3DOF embedded controller~(Measured for the ZedBoard~\cite{zboarddatasheet} used in the implementation of the controller in next section) and $10s$ 
for reboot time of the temperature controller embedded system~(we assumed platform with embedded Linux OS). The ranges for the regions and obtained availability results are presented in Table~\ref{table:availability}. 

From these results, the impact of our approach on the temperature control system is negligible. However, there is a considerable impact on the availability of the helicopter controller due to frequent restarts. Notice that, the helicopter system is among the most unstable systems and therefore, one of the most challenging ones to provide \emph{guaranteed} protection. As a result, the calculated results for the helicopter system can be considered as an approximated upper bound on the impact of our approach on controller availability among all the systems.In the next section, we use a realistic 3DOF helicopter and demonstrate that, despite the reduced the availability, the plant remains safe and stable and makes progress.

\begin{table}[]
	\centering
	\begin{tabular}{c|c|c|}
		\cline{2-3}
		& \cellcolor[HTML]{C0C0C0}\begin{tabular}[c]{@{}c@{}}Regions (From most\\ common to least Common)\end{tabular} & \cellcolor[HTML]{C0C0C0}Avail. \\ \hline
		\multicolumn{1}{|c|}{\cellcolor[HTML]{C0C0C0}}                                                                                         & 15\textless $T_O$ \textless 40                                                                                                            &                                \\ \cline{2-2}
		\multicolumn{1}{|c|}{\cellcolor[HTML]{C0C0C0}}                                                                                         & 0\textless $T_O$ \textless 15 or 40 \textless $T_O$ \textless 60                                                                                                           &                                \\ \cline{2-2}
		\multicolumn{1}{|c|}{\multirow{-3}{*}{\cellcolor[HTML]{C0C0C0}\begin{tabular}[c]{@{}c@{}}Temperature\\ Control\\ System\end{tabular}}} & $T_O$ \textless 0 or 60 \textless $T_O$                                                                                                            & \multirow{-3}{*}{\%99.9}       \\ \hline
		\multicolumn{1}{|c|}{\cellcolor[HTML]{C0C0C0}}                                                                                         & $-\epsilon + |\rho| <$  0.1 \& $\epsilon < 0.2$ \& $|\rho| < \pi/8 $                                                                                                           &                                \\ \cline{2-2}

		\multicolumn{1}{|c|}{\cellcolor[HTML]{C0C0C0}}                                                                                         & 0.2$< -\epsilon + |\rho| <$  0.1 \& $0.2<\epsilon < 0.3$                                                                                         &                                \\
		\multicolumn{1}{|c|}{\cellcolor[HTML]{C0C0C0}}                                                                                         & \& $\pi/8 < |\rho| < \pi/6$                                                                                                           &                                \\ \cline{2-2}
		\multicolumn{1}{|c|}{\multirow{-5}{*}{\cellcolor[HTML]{C0C0C0}\begin{tabular}[c]{@{}c@{}}3DOF\\  Helicopter\end{tabular}}}             & $ -\epsilon + |\rho| <$  0.2 \& $0.3<\epsilon$ \& $\pi/6 < |\rho| $                                                                                                           & \multirow{-5}{*}{\%64.3}       \\ \hline
	\end{tabular}
\caption{Weighted average availability of the embedded system.}
	\label{table:availability}
\end{table}

\section{Implementation and Attacks}

\label{sec:implementation}

Now, we present an overview of our prototype implementation of a restart-based protection approach for the control embedded system of an \emph{actual} 3DOF helicopter platforms. More detailed description on the implementation is given in our technical report~\cite{githubrepoemsoft}. In addition, through \emph{hardware-in-the-loop} approach, we use the same embedded system for controlling the warehouse temperature system. We test some attacks on both system and verify that the safety requirements are respected. 

\begin{figure}[ht]
	\begin{center}
		\includegraphics[width=0.45\textwidth]{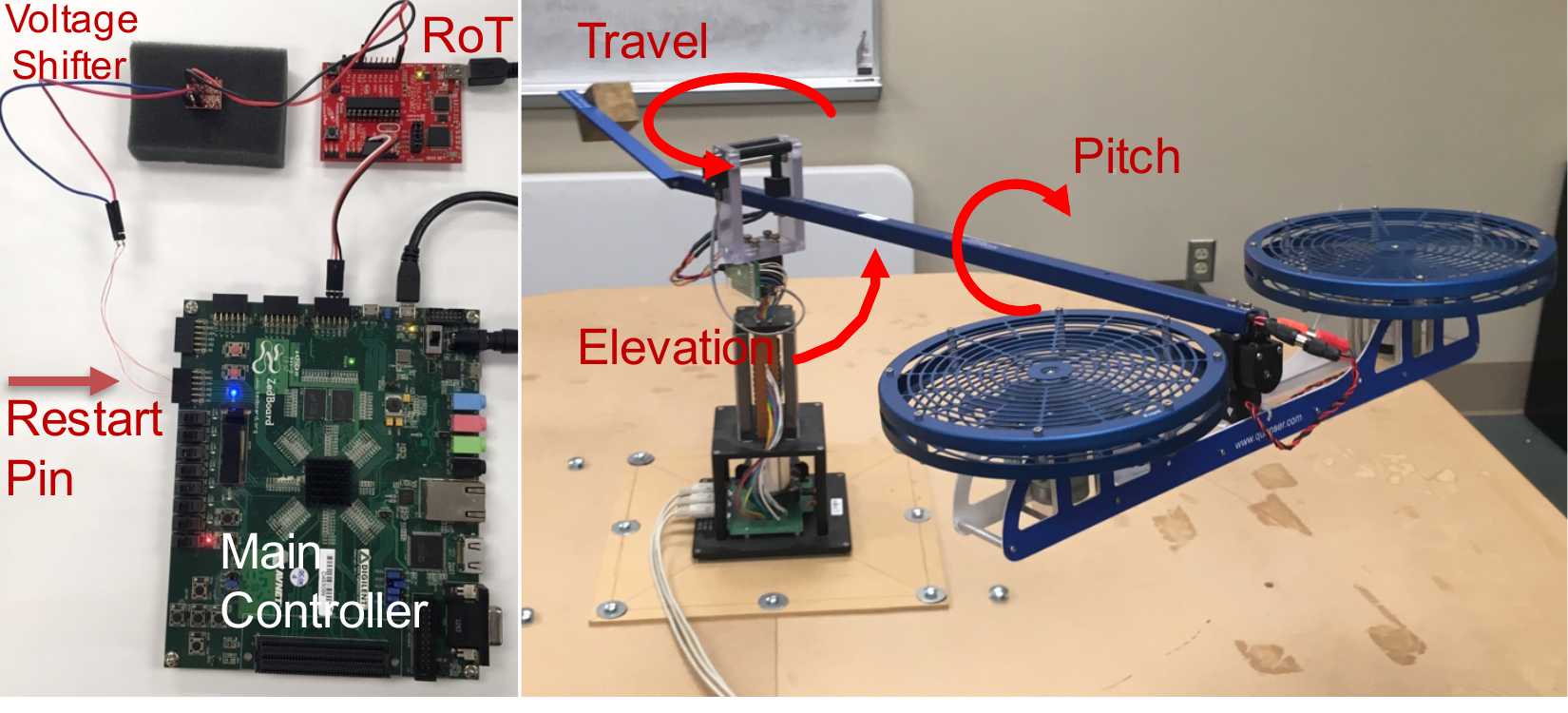}
		\caption{3DOF helicopter used as the test-bed.}
		\label{fig:3dof}
	\end{center}
\end{figure}

\subsubsection*{RoT Module:}
We implemented the RoT module using a minimal MSP430G2452 micro-controller on a MSP-EXP430G2 LaunchPad board~\cite{launchpadboard}. 
To enable restarting, pin P2.0 of the micro-controller is connected to the restart input of the main controller. Internal \emph{Timer A} of the micro-controller is used for implementing the restart timer. It is a 16-bit timer configured to run at a clock rate of 1MHz (\ie $1\mu s$ per timer count) using the internal digitally controlled oscillator. A counter inside the interrupt handler of \emph{Timer A} is used to extend the timer with an adjustable factor, in order to enable the restart timer to count up to the required range based on the application's needs.

The $I^2C$ interface is adopted for the main controller to set the restart time on RoT module. After each restart, RoT acts as an $I^2C$ slave waiting for the value of the restart time. As soon as the main controller sends the restart time, RoT disables the $I^2C$ interface and activates the internal timer. Upon expiration of the timer, an active signal is set on the restart pin to trigger the restart event and the $I^2C$ interface is activated again for accepting the next restart time.

\subsubsection*{Main Controller:}
We used a Zedboard~\cite{zboarddatasheet} to implement the main controller of the 3DOF helicopter and warehouse temperature control system. Zedboard is a development board for Xilinx's Zynq-7000 series all programmable SoC. 
It contains a XC7Z020 SoC, 512MB DDR3 memory and
an on-board 256MB QSPI Flash.
The XC7Z020 SoC consists of a processing system (PS) with dual ARM Cortex-A9 cores and a 7-series programmable logic (PL).
The processing system runs at 667MHz.
In this evaluation, only one of the ARM cores is used while the idle one is not activated.
The programmable logic is programmed to provide the $I^2C$ and \emph{UART} interfaces that are used for connecting to the RoT module and the 3DOF helicopter.

The main controller runs \emph{FreeRTOS}~\cite{FreeRTOS}, a preemptive real-time operating system. Immediately after the reboot when the \emph{FreeRTOS} starts, \texttt{SafetyController} and \texttt{FindRestartTime} tasks are created and executed. \texttt{SafetyController} is a periodic task with a period of $20ms$~($50Hz$). And, \texttt{FindRestartTime} is a single task that executes in a loop, and it only breaks out when a positive restart time is found. At this point, the restart time is sent to the RoT module via $I^2C$ interface, \texttt{SafetyController} and \texttt{FindRestartTime} tasks are terminated and the main control application tasks are created. To calculate reachability in run-time in function \texttt{FindRestartTime}, we used the implementation of real-time reachability approach in~\cite{rt-reach}.

Reset pin of Zedboard is connected to RoT module's reset output pin.
The entire system (both PS and PL) on Zedboard is restarted when the reset pin is pulled to low state.
The boot process starts when the reset pin is released (returning to high state).
A boot-loader is first loaded from the on-board QSPI Flash.
The image for PL is then loaded by the boot-loader to program the PL which is necessary for PS to operate correctly.
Once PL is ready, the image for PS is loaded and \emph{FreeRTOS} will take over the control of the system.

To further enhance the security of \emph{FreeRTOS}, we adopt \emph{randomization techniques} to randomize the following attributes: 
\ci the order of the task creation and \cii task scheduling.
After every restart, we randomize the order of creating the application tasks.
By doing so, the memory allocation is somewhat shuffled as the memory blocks are assigned to the tasks dynamically when they are created in \emph{FreeRTOS}.
Additionally, we port \emph{TaskShuffler}, a schedule randomization protocol~\cite{taskshuffler}, on \emph{FreeRTOS}.
This technique randomly selects a task from the scheduler's ready queue to execute, subject to each task's pre-computed priority inversion budget, at every scheduling point. Randomizing system parameters increases the difficulty of launching the same attack after each restart.

\textbf{3DOF Helicopter Controller:} The main controller unit interfaces with the 3DOF helicopter through a PCIe-based \textit{Q8 High-Performance H.I.L. Control and data acquisition unit}~\cite{q8daq} and an intermediate Linux-based PC. The PC communicates with the ZedBoard through the serial port. At every control cycle, a task on the controller communicates with the PC to receive the sensor readings~(elevation, pitch, and travel angles) and send the motors' voltages. The PC uses a custom Linux driver to send the voltages to the 3DOF helicopter motors and reads the sensor values. The code for the controller including the controllers and optimizations to reduce the boot sequence of ZedBoard is available online~\cite{githubrepoemsoft}.

\textbf{Warehouse Temperature Controller:} For system, due to lack of access to the real physical plant we used a hardware-in-the-loop approach. Here, the PC runs a simulation of the temperature based on the heat transfer model of Equation~\ref{eq:heatermodel1} and~\ref{eq:heatermodel2}. Similar to the 3DOF helicopter, the controller is implemented on the ZedBoard with the exact same configuration~(RoT, serial port connection, I$^2$C interface, $50Hz$ frequency) as described earlier. Control commands are sent to the PC, applied to the system and the state is reported back to the controller.

%
%

\subsection{Attacks on the Embedded System}

To evaluate the effectiveness of the proposed approach, we tested three types of  synthetic attacks on the implemented controller of the 3DOF helicopter with actual plant and one attack on the hardware-in-the-loop implementation of the warehouse temperature control system.In these experiments, our focus is on the actions of the attacker after the breach into the system has taken place. Hence, the breaching approach and exploitation of the vulnerabilities is not a concern of these experiments.

In the first attack experiment, we evaluate the protection provided by our approach in presence of an attacker who, once activated, is capable of killing the main controller task. The attack is activated at a random time after the end of SEI. We observed that under this attack, the 3DOF helicopter did not hit the surface~(\ie it always remained within the set of admissible states).

In the second attack experiment, attacker replaces the sensor readings of the system with corrupted values with the aim of destabilizing the plant and reducing its progress. The activation time of this attack, similar to the previous one, is dictated by a random variable. Similar to the first attack experiment, system remained safe throughout the attacks. And the progress of the system was negatively impacted for as the activation time of the attack became smaller.

In the third attack experiment, we investigate the impact of a \emph{worst-case} attacker who becomes active immediately after the SEI and replaces the original controller with a malicious process that turns off the fans of the helicopter leading it to hit the surface. We observed that the system was able to tolerate this attack and did not hit the surface. The trace of the system during a time interval of activity is plotted in Figure~\ref{fig:traceofattack}. 

From figure, it can be seen the controller spends most of the time in SEI~(red region) or in reboot~(white region). This is due to the fact that this extreme-case attack is activated immediately after each SEI and destabilizes the helicopter. By the time that the reboot is complete~(end of the white region), system is close to unsafe states. Hence, SEI becomes longer as SC is stabilizing the system. under this very extreme attack model, the system does not make any progress towards its designated path under this attack model.

\begin{figure}
\includegraphics[width=0.48\textwidth]{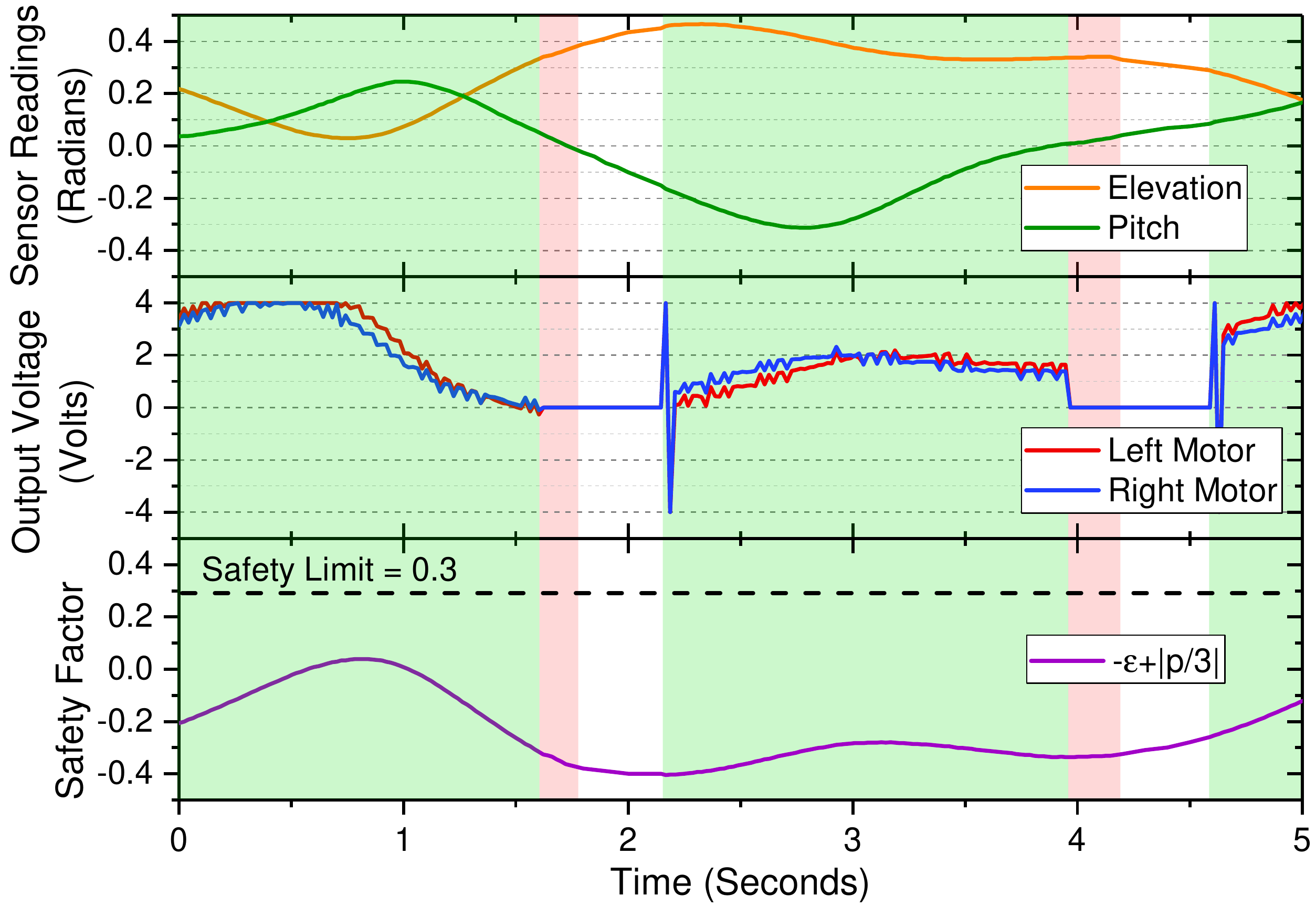}
\caption{Trace of 3DOF Helicopter during two restart cycles when the system is under worst-case attack that is active immediately after SEI. Green: SEI, red: system reboot, white: Normal Operation~(In this case attacker)}
\label{fig:traceofattack}
\end{figure}

In the last experiment, an attack is performed on the embedded controller of the warehouse temperature. In this experiment, the outside temperature is set to $45^\circ C$ and initial temperature of the room is set to $25^\circ C$. The attacker takes full control of the system and immediately after the SEI, activates both heaters to increase the temperature. We observed that system restarts before the room temperature goes above $30^\circ C$ and after the restart, SC drives the temperature towards the safe range.

\section{Discussion}
\label{sec:discussion}

\textbf{Suitable Application Targets:} Restart-based protection aims to provide a very strong protection guarantees for the CPS. However, it is not necessarily applicable to every system. The main limiting factor is the speed of the dynamics of the physical plant and the restart time of the embedded platform. Physical systems that are unstable, have very narrow safety margin, or have very fast moving dynamics have a very low tolerance for misbehavior. Adversaries can damage such systems very quickly. As a result, applying restart-based protection to such systems would require them to restart very fast. Very fast restart rates may reduce the controller's available time and essentially prevent the system from making any progress. It may also create implementation challenges; a system may require some time to establish connection over Internet or to authenticate that may not be possible if the system has to restart very frequently. Therefore, before adapting the proposed mechanism for a system, a set of tests similar to the evaluations in Section~\ref{sec:restarttimeeval} need to be performed to check if the restart times are within a suitable range. On the other hand, systems with slower dynamics and winder safety margins~(\eg Internet of Things~(IoT) applications, infrastructure, municipal systems, transportation, and \etc) are the most suitable category to benefit from this technique.


\textbf{Threat Model:} The restart-based approach proposed in this paper provides guaranteed protection for the physical sub-system in presence of the attacks. However, attackers may cause other forms of damage to the system that are not physical; \eg attacks targeting performance of the system, logging and stealing information, and man-in-the middle attacks. Even though such attacks are not the main focus of this paper, a similar approach to the analysis of Section~\ref{sec:probability} can be performed to evaluate the effectiveness of restart-based protection on mitigating the impact of such attacks.


\textbf{Increasing the Entropy:} Despite the effectiveness of restarting and software reloading in removing the malicious components from the system, restarting does not fix the vulnerabilities of the system that provided the opportunity for the intruder to breach into the system. An attacker maybe able to re-launch the same attack after the restart and gain control of the system. It turns out, randomizing parameters and operation of the system introduces a noise in the execution pattern of the system within each cycle, thereby, increasing the time and effort needed by the attacker to trigger the attack. In other words, randomization reduces the impact of previous knowledge in launching new attacks. There is a considerable body of work on various randomization techniques for security protection; \eg schedule Randomization~\cite{taskshuffler}, address space randomization~(ASLR)~\cite{Shacham:2004:EAR:1030083.1030124}, and Kernel Address Space Randomization~(KASLR)~\cite{Jang:2016:BKA:2976749.2978321}.


\section{Related Work}
\label{sec:related}

The idea of restarting the entire system or its components at run-time has been explored in earlier research in two forms of \ci \emph{revival} (\ie reactively restart a failed component) and \cii \emph{rejuvenation} (\ie prophetically restart functioning components to prevent state degradation).  
Authors in literature \cite{candea2001recursive} introduce recursively restartable systems as a design paradigm for highly available systems and uses a combination of revival and rejuvenation techniques. Earlier research~\cite{Candea03crash-onlysoftware,candea2003jagr,candea2004microreboot} illustrate the concept of microreboot that consists of having fine-grain rebootable components and trying to restart them from the smallest component to the biggest one in the presence of faults. Others in literature~\cite{vaidyanathan2005comprehensive,garg1995analysis,huang1995software} focus on failure and fault modeling and trying to find an optimal rejuvenation strategy for various systems. The above cited works are proposed for traditional (\ie \textit{non safety-critical}) computing systems such as servers and switches and are not directly applicable to safety-critical CPS. However, they demonstrate that the concept of restarting the system is considered as a reliable method to recover a faulty/compromised system.


In the context of safety-critical systems, authors proposed to utilize restarting as a recovery method from SW bugs/faults for safety-critical applications using an additional HW unit as a back up controller during the restart~\cite{fardin2016reset}. Unlike ours, the architecture is only concerned with fault-tolerance and the security aspect of the safety-critical CPS was not addressed. Other work exists~\cite{abdi2017application}  where the authors propose the procedures to design a controller (referred to as \textit{base controller}) that enable the entire computing system to be safely restarted at run-time. 
Base Controller keeps the system inside a subset of safety region by updating the actuator input at least once after every system restart. 
Similar to the previous work, this work is also concerned with the fault-tolerance and cannot provide any guarantees in presence of an adversary.
Further work~\cite{mhasan_resecure16} leverages the idea of restarting to improve security of safety-critical real-time systems. 
However, it requires an additional customized hardware unit to execute a back up controller during the restart to maintain the plant's stability. Our approach provides guaranteed safety for the plant and is implementable on single COTS embedded system.
While the idea of restarting the embedded system is explored in earlier research for fault-tolerance, to the best of our knowledge, this paper is the first comprehensive work to utilize system-wide restarts and physical dynamics of the system to provide security for safety-critical CPS.


There is a considerable overlap between real-time systems and safety-critical embedded CPS. The problem of studying various There are recent works that study the problem of attack protection and detection approaches in real-time systems. Recent work \cite{slack_cornell, securecore, securecore_memory} on dual-core based hardware/software architectural frameworks aim to protect RTS against security vulnerabilities. In literature \cite{mohan_s3a} authors leverage the deterministic execution behavior of RTS and use Simplex architecture~\cite{sha2001using} to detect intrusion while guaranteeing the safety. Despite the improved security, all these techniques have false negatives and do not provide any guarantees. In contrast, our restart-based mechanism guarantees that the attacker cannot damage the physical sub-system.
 

\section{Conclusion}
\label{sec:conclusion}
	In this paper, we aim to decouple the \textit{safety} requirements of the physical plant from 
the \textit{security} properties. 
Because of inertia, pushing 
a physical plant from a given (potentially safe) state to an unsafe state --
even with complete adversarial control -- is not instantaneous and often takes 
finite (even considerable) time. We leverage
this property to calculate a \textit{safe operational window} and combine it with the effectiveness of {\em system-wide
	restarts} to decrease the efficacy of malicious actors.
Designers of such safety-critical systems can now evaluate
the necessary trade-offs between control system performance and
increased security guarantees -- thus improving the overall design of embedded CPS in the future.


\bibliographystyle{abbrv}
\bibliography{monowar,fardin, implementation,blank}

\begin{thebibliography}{10}

\bibitem{FreeRTOS}
{FreeRTOS }.
\newblock \url{http://www.freertos.org}, 2016.
\newblock Accessed: Sep. 2016.

\bibitem{githubrepoemsoft}
\url{https://github.com/emsoft2017restart/restart-based-framework-demo}, 2017.

\bibitem{mhasan_resecure16}
F.~Abdi, M.~Hasan, S.~Mohan, D.~Agarwal, and M.~Caccamo.
\newblock {ReSecure}: A restart-based security protocol for tightly actuated
  hard real-time systems.
\newblock In {\em IEEE CERTS}, pages 47--54, 2016.

\bibitem{fardin2016reset}
F.~Abdi, R.~Mancuso, S.~Bak, O.~Dantsker, and M.~Caccamo.
\newblock Reset-based recovery for real-time cyber-physical systems with
  temporal safety constraints.
\newblock In {\em IEEE 21st Conference on Emerging Technologies Factory
  Automation (ETFA 2016)}, 2016.

\bibitem{abdi2017application}
F.~Abdi, R.~Tabish, M.~Rungger, M.~Zamani, and M.~Caccamo.
\newblock Application and system-level software fault tolerance through full
  system restarts.
\newblock In {\em In Proceedings of the 8th ACM/IEEE International Conference
  on Cyber-Physical Systems}. IEEE, 2017.

\bibitem{Arnold2014}
F.~Arnold, H.~Hermanns, R.~Pulungan, and M.~Stoelinga.
\newblock {\em Time-Dependent Analysis of Attacks}, pages 285--305.
\newblock Springer Berlin Heidelberg, Berlin, Heidelberg, 2014.

\bibitem{zboarddatasheet}
AVNET.
\newblock Zedboard hardware user’s guide.

\bibitem{rt-reach}
S.~Bak, T.~T. Johnson, M.~Caccamo, and L.~Sha.
\newblock Real-time reachability for verified simplex design.
\newblock In {\em Real-Time Systems Symposium (RTSS), 2014 IEEE}, pages
  138--148. IEEE, 2014.

\bibitem{candea2001recursive}
G.~Candea and A.~Fox.
\newblock Recursive restartability: Turning the reboot sledgehammer into a
  scalpel.
\newblock In {\em Hot Topics in Operating Systems, 2001. Proceedings of the
  Eighth Workshop on}, pages 125--130. IEEE, 2001.

\bibitem{Candea03crash-onlysoftware}
G.~Candea and A.~Fox.
\newblock Crash-only software.
\newblock In {\em HotOS IX: The 9th Workshop on Hot Topics in Operating
  Systems}, pages 67--72, 2003.

\bibitem{candea2004microreboot}
G.~Candea, S.~Kawamoto, Y.~Fujiki, G.~Friedman, and A.~Fox.
\newblock Microreboot- a technique for cheap recovery.
\newblock In {\em Proceedings of the 6th Conference on Symposium on Opearting
  Systems Design \& Implementation - Volume 6}, OSDI'04, pages 3--3, 2004.

\bibitem{candea2003jagr}
G.~Candea, E.~Kiciman, S.~Zhang, P.~Keyani, and A.~Fox.
\newblock Jagr: An autonomous self-recovering application server.
\newblock In {\em Autonomic Computing Workshop. 2003. Proceedings of the},
  pages 168--177. IEEE, 2003.

\bibitem{checkoway2011comprehensive}
S.~Checkoway, D.~McCoy, B.~Kantor, D.~Anderson, H.~Shacham, S.~Savage,
  K.~Koscher, A.~Czeskis, F.~Roesner, T.~Kohno, et~al.
\newblock Comprehensive experimental analyses of automotive attack surfaces.
\newblock In {\em USENIX Sec. Symp.}, 2011.

\bibitem{cy_side_channel}
C.-Y. Chen, R.~B. Bobba, and S.~Mohan.
\newblock Schedule-based side-channel attack in fixed-priority real-time
  systems.
\newblock Technical report, University of Illinois at Urbana Champaign,
  \url{http://hdl.handle.net/2142/88344}, 2015.
\newblock [Online].

\bibitem{garg1995analysis}
S.~Garg, A.~Puliafito, M.~Telek, and K.~S. Trivedi.
\newblock Analysis of software rejuvenation using markov regenerative
  stochastic petri net.
\newblock In {\em Software Reliability Engineering, 1995. Proceedings., Sixth
  International Symposium on}, pages 180--187. IEEE, 1995.

\bibitem{gollakota2011they}
S.~Gollakota, H.~Hassanieh, B.~Ransford, D.~Katabi, and K.~Fu.
\newblock They can hear your heartbeats: non-invasive security for implantable
  medical devices.
\newblock {\em ACM SIGCOMM Computer Communication Review}, 41(4):2--13, 2011.

\bibitem{5399279}
O.~Henniger, L.~Apvrille, A.~Fuchs, Y.~Roudier, A.~Ruddle, and B.~Weyl.
\newblock Security requirements for automotive on-board networks.
\newblock In {\em 2009 9th International Conference on Intelligent Transport
  Systems Telecommunications, (ITST)}, pages 641--646, Oct 2009.

\bibitem{huang1995software}
Y.~Huang, C.~Kintala, N.~Kolettis, and N.~D. Fulton.
\newblock Software rejuvenation: Analysis, module and applications.
\newblock In {\em Fault-Tolerant Computing, 1995. FTCS-25. Digest of Papers.,
  Twenty-Fifth International Symposium on}, pages 381--390. IEEE, 1995.

\bibitem{3dofhelicopter}
Q.~Inc.
\newblock 3 dof helicopter.

\bibitem{q8daq}
Q.~Inc.
\newblock Q8 data acquisition board.

\bibitem{launchpadboard}
T.~Instruments.
\newblock Msp-exp430g2 launchpad development kit.

\bibitem{Jang:2016:BKA:2976749.2978321}
Y.~Jang, S.~Lee, and T.~Kim.
\newblock Breaking kernel address space layout randomization with intel tsx.
\newblock In {\em Proceedings of the 2016 ACM SIGSAC Conference on Computer and
  Communications Security}, CCS '16, pages 380--392, New York, NY, USA, 2016.
  ACM.

\bibitem{ris_rts_1}
K.~Koscher, A.~Czeskis, F.~Roesner, S.~Patel, T.~Kohno, S.~Checkoway, D.~McCoy,
  B.~Kantor, D.~Anderson, H.~Shacham, et~al.
\newblock Experimental security analysis of a modern automobile.
\newblock In {\em IEEE Symposium on Security and Privacy}, pages 447--462.
  IEEE, 2010.

\bibitem{kumar2015time}
R.~Kumar, D.~Guck, and M.~Stoelinga.
\newblock Time dependent analysis with dynamic counter measure trees.
\newblock {\em arXiv preprint arXiv:1510.00050}, 2015.

\bibitem{slack_cornell}
D.~Lo, M.~Ismail, T.~Chen, and G.~E. Suh.
\newblock Slack-aware opportunistic monitoring for real-time systems.
\newblock In {\em IEEE RTAS}, pages 203--214, 2014.

\bibitem{mohan_s3a}
S.~Mohan, S.~Bak, E.~Betti, H.~Yun, L.~Sha, and M.~Caccamo.
\newblock S3a: Secure system simplex architecture for enhanced security and
  robustness of cyber-physical systems.
\newblock In {\em ACM international conference on High confidence networked
  systems}, pages 65--74. ACM, 2013.

\bibitem{3dof}
{Quanser Inc.}
\newblock {3-DOF} helicopter reference manual.
\newblock Document Number 644, Revision 2.1.

\bibitem{schneier1999attack}
B.~Schneier.
\newblock Attack trees.
\newblock {\em Dr. Dobb’s journal}, 24(12):21--29, 1999.

\bibitem{seto1999case}
D.~Seto and L.~Sha.
\newblock A case study on analytical analysis of the inverted pendulum
  real-time control system.
\newblock Technical report, DTIC Document, 1999.

\bibitem{Sha01usingsimplicity}
L.~Sha.
\newblock Using simplicity to control complexity.
\newblock pages 20--28. IEEE Software, 2001.

\bibitem{sha2001using}
L.~Sha.
\newblock Using simplicity to control complexity.
\newblock {\em IEEE Software}, 18(4):20--28, 2001.

\bibitem{Shacham:2004:EAR:1030083.1030124}
H.~Shacham, M.~Page, B.~Pfaff, E.-J. Goh, N.~Modadugu, and D.~Boneh.
\newblock On the effectiveness of address-space randomization.
\newblock In {\em Proceedings of the 11th ACM Conference on Computer and
  Communications Security}, CCS '04, pages 298--307, New York, NY, USA, 2004.
  ACM.

\bibitem{dronhack}
D.~P. Shepard, J.~A. Bhatti, T.~E. Humphreys, and A.~A. Fansler.
\newblock Evaluation of smart grid and civilian uav vulnerability to gps
  spoofing attacks.
\newblock In {\em Proceedings of the ION GNSS Meeting}, volume~3, 2012.

\bibitem{trapnes2013optimal}
S.~H. Trapnes.
\newblock Optimal temperature control of rooms for minimum energy cost, 2013.

\bibitem{5751382}
A.~Ukil, J.~Sen, and S.~Koilakonda.
\newblock Embedded security for internet of things.
\newblock In {\em 2011 2nd National Conference on Emerging Trends and
  Applications in Computer Science}, pages 1--6, March 2011.

\bibitem{vaidyanathan2005comprehensive}
K.~Vaidyanathan and K.~S. Trivedi.
\newblock A comprehensive model for software rejuvenation.
\newblock {\em Dependable and Secure Computing, IEEE Transactions on},
  2(2):124--137, 2005.

\bibitem{taskshuffler}
M.-K. Yoon, S.~Mohan, C.-Y. Chen, and L.~Sha.
\newblock {TaskShuffler}: A schedule randomization protocol for obfuscation
  against timing inference attacks in real-time systems.
\newblock In {\em IEEE Real-Time and Embedded Technology and Applications
  Symposium (RTAS)}, pages 1--12. IEEE, 2016.

\bibitem{securecore}
M.-K. Yoon, S.~Mohan, J.~Choi, J.-E. Kim, and L.~Sha.
\newblock Securecore: A multicore-based intrusion detection architecture for
  real-time embedded systems.
\newblock In {\em IEEE Real-Time and Embedded Technology and Applications
  Symposium (RTAS)}, pages 21--32. IEEE, 2013.

\bibitem{securecore_memory}
M.-K. Yoon, S.~Mohan, J.~Choi, and L.~Sha.
\newblock Memory heat map: anomaly detection in real-time embedded systems
  using memory behavior.
\newblock In {\em ACM/EDAC/IEEE Design Automation Conference (DAC)}, pages
  1--6. IEEE, 2015.

\end{thebibliography}

\end{document}